\documentclass[sigplan,10pt,nonacm]{acmart}

\usepackage{tikz}
\usepackage{amsmath}
\usepackage{amsthm}
\usepackage{graphicx}
\usepackage{amssymb}
\usepackage{booktabs}
\usepackage{threeparttable}
\usepackage{pgfplots}
\pgfplotsset{compat=1.18}
\usepgfplotslibrary{fillbetween}
\usetikzlibrary{shapes.geometric, arrows.meta, positioning, calc, fit, backgrounds, shadows, shapes.symbols, matrix, decorations.pathreplacing}

\renewcommand\footnotetextcopyrightpermission[1]{}

\definecolor{amdgreen}{HTML}{D6EDD0}
\definecolor{tagcabi}{HTML}{FBD0C4}
\definecolor{tagimproc}{HTML}{DCC6EE}
\definecolor{tagcodecin}{HTML}{FAEAB0}
\definecolor{tagcodecfree}{HTML}{E2E2E2}
\newcommand{\mtag}[2]{\colorbox{#1}{\textbf{\vphantom{Adjy}#2}}}
\newcommand{\breakmark}[2]{%
  \begin{scope}[shift={(#1)}]
    \fill[white] (-#2,-0.10) sin (-#2/2,-0.058) cos (0,-0.10) sin (#2/2,-0.142) cos (#2,-0.10)
      -- (#2,0.10) sin (#2/2,0.058) cos (0,0.10) sin (-#2/2,0.142) cos (-#2,0.10) -- cycle;
    \draw[black!60, line width=0.5pt] (-#2,-0.10) sin (-#2/2,-0.058) cos (0,-0.10) sin (#2/2,-0.142) cos (#2,-0.10);
    \draw[black!60, line width=0.5pt] (-#2,0.10) sin (-#2/2,0.142) cos (0,0.10) sin (#2/2,0.058) cos (#2,0.10);
  \end{scope}}

\begin{document}
% (--) -- (--) -- (--) -- (--) -- (--) -- (--) -- (--) -- (--) -

\title{The World's Fastest Matching Engine Algorithm}

\author{Jake Yoon}
\email{jake@flash1.com}
\affiliation{
  \institution{Flash One Technologies LLC}
  \country{United States}
}

\begin{abstract}
We drove \textbf{247} matching engines through one C~ABI harness on one identical workload: every
open-source FIFO implementation we could find, deduplicated, and our own, on the same gate. The
workload doubles as a byte-identical correctness oracle, \textbf{1{,}000{,}000{,}000+ order
messages} per engine, replayed against an independent-engine consensus. Only \textbf{47} are
correct as shipped; we filed \textbf{181} GitHub issues upstream, \textbf{28} already fixed by
their maintainers, none declined. Our engine leads the \textbf{160} that reproduce the consensus by
\textbf{${\sim}$95~M/s} (\textbf{12.6$\times$} the second best) on worst-case throughput. One core
sustains \textbf{103.4~million order messages per second} (122.09~million on AMD~EPYC processors) worst-case, and in the engine's production configuration its wire pipeline keeps single-book P99 host-path latency, OUCH parsing and OUCH/ITCH encoding included, \textbf{under a microsecond through 82~M~msgs/s}; a
96-core server (${\sim}$\$1{,}630/month, 3-year reserved) reaches \textbf{${\sim}\,$1.3~billion/s} across
over 10{,}000 symbols, for scale, over \textbf{45$\times$} the CTA consolidated quote feed's provisioned capacity. The lead is structural: the \textbf{73} engines written inside the trading industry sit under the
same \textbf{8.19~M/s} ceiling as the rest of the field, and it is one of them that sets it. 
\end{abstract}

\maketitle

% (--) -- (--) -- (--) -- (--) -- (--) -- (--) -- (--) -- (--) -
\section{Introduction}
\label{sec:intro}

A single CPU core in our matching engine sustains \textbf{103.4~million order messages per second} (122.09~million on AMD~EPYC processors) worst-case, and at an 82~M~msgs/s offered load (over two-thirds of its clean operating edge) the engine's wire pipeline holds a sub-microsecond single-book P99 end-to-end host-path response latency, OUCH parsing and OUCH/ITCH encoding included. On a public benchmark that drives \textbf{247} matching engines through identical work, it leads the \textbf{160} that reproduce a byte-identical correctness oracle by \textbf{${\sim}$95~M/s} (roughly $12.6\times$) on worst-case throughput (Figure~\ref{fig:cliff}), and outruns the fastest widely-used and professionally-authored builds, a 1{,}041-star project (CppTrader) and many other trading-industry engineers' order books. Scaled out, one commodity 96-core server (${\sim}$\$1{,}630/month, 3-year reserved) reaches \textbf{${\sim}\,$1.3~billion} messages per second across over 10{,}000 symbols: for scale, over \textbf{45$\times$} the provisioned capacity of the CTA consolidated quote feed (a market-data workload, not order matching).

\begin{figure*}[t]
    \centering
    \definecolor{deepred}{HTML}{A11E2C}
    \definecolor{barlight}{HTML}{C7CCD4}
    \definecolor{barhi}{HTML}{9AA0AC}
    \definecolor{pedcol}{HTML}{394150}
    \begin{tikzpicture}[font=\sffamily]
    \begin{axis}[
        width=\textwidth, height=6.9cm,
        axis y line=left, axis x line=none,
        y axis line style={black!30, line width=0.5pt},
        ylabel={Worst-case Throughput (M\,msgs/s)},
        label style={font=\small, text=black!60},
        xmin=-3, xmax=170, ymin=0, ymax=33.2,
        ytick={0,8.19,33.2}, yticklabels={0,8.2,103.4},
        tick style={black!30}, tick label style={font=\small, text=black!55},
        clip=false,
    ]
        %, the empty band above the field's ceiling,
        \fill[deepred!6] (axis cs:-3,8.19) rectangle (axis cs:160.4,33.2);
        \draw[black!30, dashed, line width=0.5pt] (axis cs:-3,8.19) -- (axis cs:160.6,8.19);
        % field: 159 open-source conforming engines as a rising histogram (light); Flash One is the 160th, drawn separately
        \addplot[ybar, bar width=2.3pt, draw=none, fill=barlight] table[x=rank,y=tput] {field_worstcase.dat};
        \fill[barhi] (axis cs:51.55,0) rectangle (axis cs:52.45,0.18);
        \fill[barhi] (axis cs:83.55,0) rectangle (axis cs:84.45,0.86);
        \fill[barhi] (axis cs:90.55,0) rectangle (axis cs:91.45,1.03);
        \fill[barhi] (axis cs:113.55,0) rectangle (axis cs:114.45,1.76);
        \fill[barhi] (axis cs:122.55,0) rectangle (axis cs:123.45,2.15);
        \fill[barhi] (axis cs:143.55,0) rectangle (axis cs:144.45,5.36);
        \fill[barhi] (axis cs:151.55,0) rectangle (axis cs:152.45,6.825);
        \fill[barhi] (axis cs:158.55,0) rectangle (axis cs:159.45,8.19);
        % ~95 M/s performance gap (left); the arrow is broken like the bar it measures
        \draw[deepred, line width=0.8pt] (axis cs:5,8.19) -- (axis cs:11,8.19);
        \draw[deepred, line width=0.8pt] (axis cs:5,33.2) -- (axis cs:11,33.2);
        \draw[deepred, line width=1pt, {Stealth[length=5pt]}-{Stealth[length=5pt]}] (axis cs:8,8.5) -- (axis cs:8,32.9);
        \node[anchor=west, font=\large\bfseries, text=deepred] at (axis cs:11.5,22) {${\sim}$95\,M/s};
        \node[anchor=west, font=\footnotesize\bfseries, text=deepred] at (axis cs:11.5,19.7) {PERFORMANCE GAP};
        % the empty band label (vertically centered in the band: (8.19+33.2)/2)
        \node[font=\large\bfseries\itshape, text=black!60] at (axis cs:82,20.7) {No other engine reaches this region.};
        % Flash One, the red bar with a white star on top; broken (tilde) because it runs off this scale to 103.4
        \fill[deepred] (axis cs:160.6,0) rectangle (axis cs:163.6,33.2);
        \node[star, star points=5, star point ratio=2.3, minimum size=12pt, inner sep=0pt, fill=white, draw=deepred, line width=1pt] at (axis cs:162.1,33.2) {};
        \node[anchor=south, font=\footnotesize\bfseries, text=deepred, align=center] at (axis cs:162.1,34.9) {FLASH ONE\\[0pt]{\large 103.4}\,M/s};
        %, pedigree tags BELOW the axis, thin leaders up to the highlighted bars,
        \node[align=center, font=\footnotesize, text=pedcol] (tHF) at (axis cs:75,-2.9) {\textbf{Hedge fund}\\[-1.5pt]quant};
        \draw[pedcol!55, line width=0.4pt] (axis cs:84,0.86) -- (tHF.north); \fill[pedcol] (axis cs:84,0.86) circle (0.7pt);
        \node[align=center, font=\footnotesize, text=pedcol] (tBANK) at (axis cs:88,-6.3) {\textbf{Global bank}\\[-1.5pt]engineer};
        \draw[pedcol!55, line width=0.4pt] (axis cs:91,1.03) -- (tBANK.north); \fill[pedcol] (axis cs:91,1.03) circle (0.7pt);
        \node[align=center, font=\footnotesize, text=pedcol] (tEXG) at (axis cs:105,-2.9) {\textbf{Major exchange}\\[-1.5pt]engineer};
        \draw[pedcol!55, line width=0.4pt] (axis cs:114,1.76) -- (tEXG.north); \fill[pedcol] (axis cs:114,1.76) circle (0.7pt);
        \node[align=center, font=\footnotesize, text=pedcol] (tMMa) at (axis cs:121,-6.3) {\textbf{Market maker}\\[-1.5pt]engineer};
        \draw[pedcol!55, line width=0.4pt] (axis cs:123,2.15) -- (tMMa.north); \fill[pedcol] (axis cs:123,2.15) circle (0.7pt);
        \node[align=center, font=\footnotesize, text=pedcol] (tMMb) at (axis cs:136,-2.9) {\textbf{Market maker}\\[-1.5pt]engineer};
        \draw[pedcol!55, line width=0.4pt] (axis cs:144,5.36) -- (tMMb.north); \fill[pedcol] (axis cs:144,5.36) circle (0.7pt);
        \node[align=center, font=\footnotesize, text=pedcol] (tEXMM) at (axis cs:155,-6.3) {\textbf{Ex-market-maker}\\[-1.5pt]engineer};
        \draw[pedcol!55, line width=0.4pt] (axis cs:152,6.825) -- (tEXMM.north); \fill[pedcol] (axis cs:152,6.825) circle (0.7pt);
        \node[font=\small, text=deepred, anchor=east, align=center] (tTop) at ([yshift=9pt]axis cs:158.6,11.4) {\textbf{8.19}\,M/s\\[-1.5pt]Market-maker engineer};
        \draw[deepred!55, line width=0.4pt] (axis cs:159,8.19) -- (tTop.south); \fill[deepred] (axis cs:159,8.19) circle (0.9pt);
        % x-axis caption
        \node[anchor=north, font=\small\bfseries, text=black!75] at (axis cs:80,-8.7) {160+ PROVEN MATCHING ENGINES};
        % one scale break, at one height, across everything it interrupts: the axis, the gap arrow, the bar
        \coordinate (brkAxis) at (axis cs:-3,20.7);
        \coordinate (brkArrow) at (axis cs:8,20.7);
        \coordinate (brkBar) at (axis cs:162.1,20.7);
    \end{axis}
    \breakmark{brkAxis}{0.15}
    \breakmark{brkArrow}{0.14}
    \breakmark{brkBar}{0.27}
    \end{tikzpicture}
    \caption{\mtag{blue!12}{ARM Graviton4} \mtag{tagcabi}{C~ABI} \textbf{Worst-case throughput of the 159 conforming open-source matching engines}
    (byte-identical to the consensus), sorted low to
    high, against Flash One (the 160th conforming engine, the red bar) on the identical workload. The field's fastest reaches
    8.19~M/s; Flash One sustains 103.4~M/s (the scale breaks midway so the field stays legible), a ${\sim}95$~M/s gap ($12.6\times$) that no engine
    enters. The highlighted bars are personal projects by engineers who identify as working, or having worked, at
    leading trading firms, banks, and exchanges (affiliations self-stated, not their employers' official engines); the fastest of all, which sets that ceiling, is itself one of them, a market-maker engineer's. A margin this size is unreachable by any
    accumulation of implementation tweaks (in HFT, trimming 2\% from the critical path is already ``a big win''~\cite{hrt2024devirt}), so it
    is the signature of a different \emph{architecture}, not a
    faster implementation of the same one.}
    \label{fig:cliff}
\end{figure*}

Modern electronic venues routinely experience micro-bursts: short, extremely intense spikes in order flow that last only microseconds yet carry a significant fraction of daily volume~\cite{menkveld2018electron,db2016xetraInsights,db2025dynamics}. During these bursts, the per-symbol matching core saturates, queues build up, and tail latency spikes and turns unpredictable~\cite{db2025dynamics,db2016xetraInsights}. Market makers react with defensive quoting, wider spreads, and reduced displayed liquidity, harming both traders and the exchange~\cite{egginton2016quotestuffing}.

Existing high-performance exchanges are already close to the limits of conventional software architectures~\cite{teverovski2023t7latency}. Public documentation reports peak throughputs of only ${\sim}300{,}000$ messages/s per partition (each bundling multiple products) on carefully engineered systems~\cite{eurex2024tech,db2025dynamics}. Because per-symbol matching logic is strictly serialized, Amdahl's law makes it the dominant bottleneck regardless of how aggressively surrounding infrastructure is parallelized or scaled out~\cite{amdahl67,db2016xetraInsights}.

This paper presents a new matching engine architecture that targets deterministic, micro-burst-resilient performance. The key observation is that most hot-path work consists of structured, cache-sensitive operations on queues of resting orders and price levels. The architecture therefore organizes the book as a hierarchy of fixed-capacity \emph{Priority-Indicated Nodes} (PINs) with contiguously addressable slots, coupled with a price-level index; its embodiments (\S\ref{sec:design}) add bounded relocation cascades, depth-aware node capacities, and a neighbor-aware balanced search tree over price levels that supports constant-time splice/graft operations from known neighbors followed by a short rebalancing walk.

We evaluate a production embodiment of this architecture on commodity CPUs, on workloads calibrated to regulator-reported market activity, including synthetic bursts of millions of back-to-back messages; the throughput and latency results summarized above are established in \S\ref{sec:evaluation}. The same architecture is built from the ground up for hardware acceleration: the PIN's fixed-capacity slot regions map directly to on-chip block RAM (BRAM), its priority indicators to hardware priority encoders, and neighbor-aware tree operations eliminate the deep data-dependent traversals hostile to hardware pipelines. A report on an FPGA realization of this specified design is underway; this paper reports the CPU results.

In summary, this paper makes the following contributions:

\begin{itemize}
  \item We formulate the micro-burst matching problem for order books and show that realistic exchange workloads and public throughput ceilings place the hot matching loop at the architectural bottleneck~\cite{db2025dynamics,db2016xetraInsights}.
  \item We introduce the \emph{Priority-Indicated Node} (PIN), a new priority queue design with (i) a contiguously
addressable region of $C$ logical slots and (ii) priority indicators encoding the
entry's global priority status, which together provide worst-case constant work per node while preserving the book's exact priority order.
  \item We develop a neighbor-aware balanced search tree framework for price levels that supports constant-time splice and graft operations from known neighbors, followed by a single root-to-leaf rebalancing path.
  \item We open-source the benchmark harness\footnote{\url{https://github.com/flash1-dev/matching-engine-benchmark}}, baseline-engine adapters, a deterministic workload generator, and byte-identical reference output hashes, so the workload, the baseline results, and the correctness oracle are independently reproducible, and any engine can be measured on the identical order stream. The same harness has since driven \textbf{247} matching engines through this workload; \textbf{160} reproduce the byte-identical consensus (no other engine reaching our worst-case throughput), only \textbf{47} are pristine as shipped, and we filed \textbf{181 GitHub issues upstream}, \textbf{28} already fixed by their maintainers, none declined (Figure~\ref{fig:waffle}; \S\ref{sec:audited_field}). Our own engine is reproducible under a separate license; its single-symbol numbers are reported here, validated against that public oracle, with the instance-level aggregate (\S\ref{sec:instance_throughput}) given as the median of ten runs.
\end{itemize}

\begin{figure*}[t]
    \centering
    \definecolor{wpristine}{HTML}{2E7D32}
    \definecolor{wfix}{HTML}{E8A33D}
    \definecolor{wbroken}{HTML}{9AA0A6}
    \definecolor{wxred}{HTML}{C62828}
    \pgfmathsetlengthmacro{\wunit}{\textwidth/28}%
    \begin{tikzpicture}[x=\wunit,y=\wunit]
      \foreach \i in {0,...,246}{%
        \pgfmathtruncatemacro{\col}{mod(\i,28)}%
        \pgfmathtruncatemacro{\row}{div(\i,28)}%
        \pgfmathtruncatemacro{\cat}{\i<47 ? 0 : (\i<160 ? 1 : 2)}%
        \edef\cc{\ifcase\cat wpristine\or wfix\or wbroken\fi}%
        \fill[\cc] (\col,-\row) rectangle ++(0.84,0.84);%
        \ifnum\cat=2%
          \draw[wxred, line width=1.1pt, line cap=round] (\col+0.14,-\row+0.14) -- (\col+0.70,-\row+0.70) (\col+0.14,-\row+0.70) -- (\col+0.70,-\row+0.14);%
        \fi%
      }
    \end{tikzpicture}\\[7pt]
    {\large
      $\vcenter{\hbox{\tikz\fill[wpristine](0,0)rectangle(2.2ex,2.2ex);}}$~\textbf{47} pristine as shipped\quad
      $\vcenter{\hbox{\tikz\fill[wfix](0,0)rectangle(2.2ex,2.2ex);}}$~\textbf{113} conform after our suggested fix\quad
      $\vcenter{\hbox{\tikz{\fill[wbroken](0,0)rectangle(2.2ex,2.2ex);\draw[wxred,line width=0.9pt,line cap=round](0.38ex,0.38ex)--(1.82ex,1.82ex) (0.38ex,1.82ex)--(1.82ex,0.38ex);}}}$~\textbf{87} non-conforming}
    \caption{\textbf{Conformance across the 247 surveyed engines.} Each is one cell, colored
    by how it fares against the byte-identical consensus: only \textbf{47} are pristine as
    shipped, \textbf{113} reach the consensus after our suggested correctness fix (a
    minimal, localized patch we drafted and filed upstream wherever the repository accepts issues, never a codebase-wide rewrite), and
    \textbf{87} are non-conforming: we found no such patch for them, so reaching the consensus
    would take a structural change or a large edit to the codebase (they have diverged, crashed,
    or failed to finish their worst scenario, a small minority by deliberate design rather than defect).
    Fewer than one in five is correct out of the box. Across the 247, we filed 181 GitHub issues
    upstream; 28 are already fixed by their maintainers, and none declined.}
    \label{fig:waffle}
\end{figure*}

% (--) -- (--) -- (--) -- (--) -- (--) -- (--) -- (--) -- (--) -
\section{Background and Motivation}
\label{sec:background}
% (--) -- (--) -- (--) -- (--) -- (--) -- (--) -- (--) -- (--) -

Electronic limit order books (LOBs) implement continuous double auctions for individual
instruments. For each symbol, the exchange maintains a price-indexed set of resting
orders and a \emph{matching engine} that consumes a totally ordered message stream,
updates the in-memory book, and emits executions and acknowledgments.

Strict price--time priority is inherently \emph{serial} at the symbol level: competing
messages for the same instrument must be processed in a single deterministic sequence.
In a continuous market this is a correctness requirement, not an implementation choice, and it
creates a hard ceiling on per-symbol throughput regardless of how aggressively other pipeline stages
are parallelized~\cite{budish2015batchauctions}. Modern exchange
architectures therefore scale \emph{across} symbols (many matching instances in parallel)
but remain throughput-limited \emph{within} a hot symbol by the performance and memory
behavior of a matching core.

\subsection{Market Micro-Bursts Expose Throughput Ceilings}
\label{subsec:microbursts}

Modern electronic markets exhibit highly bursty order-arrival patterns: long periods of
moderate activity punctuated by sub-millisecond \emph{micro-bursts} triggered by news or
auction transitions. A disproportionate share of trading occurs in these short intervals:
about a fifth of trades fall in sub-millisecond clusters~\cite{menkveld2018electron}.
The dynamic is self-reinforcing: the faster a venue runs, the faster participants react and
the sharper the next burst, a feedback loop Deutsche B\"orse describes between its own
latency and customer reaction time~\cite{teverovski2023t7latency}.

From a systems perspective, a micro-burst is the regime where the instantaneous ingress
rate for a single symbol exceeds the sustainable service rate of its matching engine.
Deutsche B\"orse's T7 measurements illustrate this gap: during one example burst,
inbound traffic peaks around \(8\) million messages/second at an early gateway timestamp,
but only ${\sim}300{,}000$ messages/second at the start of matching, with intermediate
stages in the few-hundred-kHz range~\cite{db2025dynamics}. Because the matcher is the
pipeline's narrowest stage, a burst that exceeds its ceiling ends up queued in front of it,
however far the gateways ahead of it scale out, and latency comes to be dominated by
\emph{queueing delay}. What protects a
venue is therefore throughput \emph{headroom}: how large a burst it can clear before a queue
forms.

\subsection{Latency Spikes Create Execution Uncertainty and Widen Spreads}
\label{subsec:market_quality_chain}

Queueing during micro-bursts matters economically because it makes execution timing
\emph{unpredictable}. A market maker's cancel/replace message competes for the same
serialized matching bandwidth as every other message; when the book is congested, stale
quotes cannot be withdrawn quickly, and rational liquidity providers respond by quoting
more conservatively (less depth, wider spreads, more frequent quote fading), raising
transaction costs for investors. Where venues compete for order flow, that congestion
compounds into a \emph{positional} disadvantage (\S\ref{subsec:commercial}): the venue
that falls behind does not just quote worse; it cedes order flow to the venue that does not.

\subsection{Why Throughput, Not Baseline Latency, Is the Remaining Bottleneck}
\label{subsec:throughput_vs_latency}

Over the last decade, venues and participants have largely exhausted low-hanging fruit in
\emph{baseline} latency through co-location, optimized networking stacks, and specialized
hardware. 
Deutsche B\"orse, for example, reports microsecond-scale baseline latencies in
T7 (a few microseconds one-way in co-located settings and low double-digit microseconds
from NIC to matching-engine ingress)~\cite{db2025dynamics}. At these scales, shaving a
few more microseconds off median latency does little to address the dominant source of
tail latency under stress: queues at the serialized match loop during bursts.

Empirical evidence also shows diminishing returns to latency reduction once fast access
is already available. A case study finds that a latency reduction at an Australian exchange
yielded some liquidity improvements but no persistent reduction in bid--ask spreads when
institutional traders already had co-location~\cite{murray2016latency}. By contrast,
exchange technology upgrades that increase \emph{capacity} have clearer
market-quality effects on both spreads and fairness; we quantify one such
upgrade (TSE's Arrowhead renewal, which doubled capacity as it halved order latency) in detail in
Section~\ref{subsec:commercial}.

\section{Data Structures and Core Operations}
\label{sec:design}

\subsection{Key Limitations of Existing Work}
Most open-source matching engines, like the textbook design, use pointer-chasing structures: each price level is a linked list (or tree of lists) of orders, with a balanced tree or array indexing price levels. This gives simple $O(1)$ queue edits and $O(\log N)$ price lookup, but it fights modern CPUs: heap-allocated nodes are scattered in memory, hardware prefetchers cannot predict pointer chains, and each list/tree step risks a dependent cache miss, so under micro-bursts the core stalls on memory rather than arithmetic. The obvious opposite extreme (storing a per-price queue as a single contiguous array) plays nicely with caches and prefetchers, but is mismatched to real workloads where roughly 97\% of orders cancel without executing~\cite{secTradeToOrderVolume2022,khomynPutnins2021} and many cancels hit in the middle of the FIFO. In an array, deleting from the middle requires shifting a suffix of the queue, giving $O(n)$ work per cancel at that level; under heavy churn, the engine burns cycles on large memmoves and loses the benefit of locality. In short, classic pointer-based designs underutilize the cache hierarchy, while naive flat arrays make random in-queue deletions prohibitively expensive; our design is motivated by the need to keep storage contiguous and prefetch-friendly without incurring $O(n)$ compaction on the dominant cancel path.

\subsection{Priority-Indicated Node (PIN)}

A \emph{Priority-Indicated Node} (PIN) is a fixed-capacity priority queue node with (i) a contiguously
addressable region of $C$ logical slots and (ii) priority indicators encoding the
entry's global priority status.

\paragraph{Contiguously addressable slot region.}
A node exposes a logical slot index space $\{0,\dots,C-1\}$ such that slot $i$ is found by
base-plus-stride arithmetic, with no pointer chasing. The region may be (i) a single
embedded array, (ii) a small number of back-to-back sub-arrays, or (iii) a contiguous
section of a shared arena owned by the node (with linear or ring indexing). Across all
cases we enforce a \emph{base/stride invariant}: during normal operations
(e.g., insert, delete, modify, and cascades), the base address (or arena offset) and
per-slot stride do not change, so consecutive indices always map to predictable
fixed-offset locations. Each slot stores either the order object or a compact reference to
an order stored elsewhere.

\paragraph{Priority indicators.}
A priority indicator encodes the \emph{global} priority status of an order
under the active rule. In a price--time book, for example, an indicator might record
that slot $i$ holds the book-wide best (or $k$-th best) order at its price
level; this is not merely a node-local ranking but a projection of the
order's position in the full priority sequence onto a slot-local
representation. In one exemplary embodiment, the node can maintain one
indicator per slot, updated whenever orders move so that each slot's
indicator stays consistent with the order it holds; more generally, a
priority indicator need not be materialized as one datum per slot. In a
sparse encoding, an indicator that is inactive or non-asserted under the
current rule (even for a slot that stores a live order) may be absent
from the encoding or set to a neutral value, avoiding overhead for
materializing indicators whose priority status is not currently relevant.

\paragraph{Representation and placement.}
Indicators are representation and placement neutral: they may be bitmasks (one bit per
slot), slot-indexed arrays of flags or references, or fields embedded in the order
objects; they may live inside the node, in a per-node external structure, or in the
orders. Empty slots are encoded by absence of an indicator entry or a neutral value
(e.g., bit=0/null). More generally, we use \emph{priority indicator} for any such
encoding; a one-indicator-per-slot layout is one realization, not a
defining restriction.

\paragraph{Why these two properties matter.}
Base/stride addressing plus priority indicators decouples logical priority from physical
layout. Node-local operations are $O(C)$ (with small $C \ll n$ total orders), reduce slot
access to base-plus-stride arithmetic, and avoid the $O(n)$ data shifting of conventional array-based designs and the
in-node compaction of unrolled lists, while retaining
cache-friendly contiguous storage.

\paragraph{Insertion and relocation.}
Every insertion is an \emph{Append} or \emph{Prepend} relative to a rule-defined reference at its price level (under strict price--time, the current tail or head): one payload write and a constant number of indicator updates. If the target node is full, a directed relocation \emph{cascade} moves one order per hop to the adjacent node toward the tail (\emph{Push Back}) or the head (\emph{Push Forward}); cascades are capped at $D_{\max}$ hops, after which a node is allocated or reused at the boundary, so insertion into a full node costs at most $O(D_{\max})$.

\subsection{Flexible Node Capacity Model}
\label{subsec:flexcap}

Node capacity need not be uniform across the book. Orders near the top of book are
accessed far more often than those in the tail, so using the same width everywhere
wastes cache and TLB footprint on cold regions while under-amortizing misses on the
hot prefix. We therefore define a \emph{Flexible Node Capacity} policy that chooses each
node's slot capacity as a function of its depth $d$ from the book head, $\kappa(d) = C_d$, fixed for the node's lifetime and subject to three constraints: capacities are monotone nonincreasing with depth ($C_0 \ge C_1 \ge \dots$, so hotter regions may be wider); every $C_d$ is bounded by a $C_{\max}$ chosen so that all per-slot indicators fit in a small, fixed number of machine words; and $\sum_d C_d = \infty$, so the book can grow arbitrarily deep even if tail nodes use minimal capacities. Changing $\kappa(\cdot)$ never forces a global reorganization; only newly allocated or recycled nodes adopt new capacities.

A simple latency model guides the choice. If in-node work costs $A$ per slot, a miss in L1 adds $t_R$, and order-flow intensity decays away from the best quotes as empirical studies of limit order books report~\cite{gould2013lob,zovkofarmer2002patience,bouchaud2002orderbook,cont2010stochastic,munitoke2013queueing}, a node whose orders are each hit with probability $C$ minimizes its expected latency at
\[
k^{\ast} = \frac{1}{C}\ln\!\left(\frac{t_R\,C}{A}\right),
\]
positive when $t_R C > A$. For cold nodes, where $t_R C \le A$, $k^{\ast}$ is non-positive and the smallest feasible capacity is optimal; hotter nodes take wider ones, with $k^{\ast}$ rising with $C$ while $t_R C/A < e$. The analytic $k^{\ast}$ is a guide rather than an exact optimum (resting depth often peaks a tick or two behind the touch); in production an online estimator (not detailed here) can tighten it from live telemetry.

\subsection{Price-Level Index with Neighbor-Aware Balanced-Tree Updates}
\label{subsec:price_index}

Each side of the order book maintains a dynamically changing set of \emph{active price
levels}; a balanced search tree over prices supports fast best-price queries and
predecessor/successor navigation needed by continuous matching.  Each tree element is a
\emph{Price Level Descriptor} that stores fixed-size metadata (price, aggregated size,
and references to the current head/tail order locations in Priority-Indicated Nodes), along
with tree links and explicit in-order neighbor links.

\paragraph{Why neighbor-aware updates.}
In conventional balanced trees, inserting or deleting a price-level key costs
$O(\log n)$ for a root-to-leaf \emph{search} (to find the structural edit location),
plus $O(\log n)$ for the subsequent \emph{fix-up} walk (rotations/splits/merges) that
restores balance~\cite{clrs2022,guibas1978dichromatic,bayer1972btree}.  The key idea of
\emph{neighbor-aware} insertion/deletion is to bypass the search phase when the
matching engine already knows the in-order predecessor and/or successor of the affected
price level.  Given these neighbors, the tree edit reduces to a constant-time splice
(or graft) localized to a small set of nodes/pages, followed by the standard fix-up
phase along a single ancestor path.  Eliminating the search traversal removes a large
fraction of pointer-chasing reads and unpredictable branches from the critical path,
which is precisely the behavior that becomes fragile under micro-bursts.

\paragraph{How the matching engine supplies neighbors.}
The matching engine naturally maintains (or can obtain at negligible marginal cost)
the neighbor information required for these updates:

\begin{itemize}
  \item \textbf{When activating a new price level.}
        A new level $p$ is created only when an incoming limit order targets a price
        with no existing resting interest.  At that moment, the engine can determine
        the immediate predecessor/successor levels in price order from \emph{state it
        already touches}:
        (i) best-price pointers and the in-order neighbor links of nearby active
        levels (common when $p$ is close to the top of book), or
        (ii) a single predecessor/successor query (e.g., ``floor'' or ``ceiling'') that
        is needed anyway to decide where $p$ sits relative to current book state.
        Crucially, once the engine has identified the bracketing levels $(P,S)$, it
        does \emph{not} re-traverse the tree from the root to locate the insertion
        point; it splices the new descriptor directly between $P$ and $S$.

  \item \textbf{When deleting an empty price level.}
        A level is deleted exactly when the last resting order at that price is
        executed or canceled.  Because each order slot ultimately maps to its owning
        level descriptor (directly or indirectly via Priority-Indicated Node metadata), the
        engine holds a pointer to the descriptor being removed.  The descriptor
        itself maintains explicit in-order neighbor links (\texttt{pred}/\texttt{succ}),
        so the engine can select the successor (or predecessor at the boundary) as a
        graft candidate without any tree search.
\end{itemize}

As a result, \emph{tree search is not used to discover neighbors}; neighbor discovery
is coupled to the matching workflow and book-local metadata, and the balanced tree is
used primarily for (i) maintaining global price order under churn and (ii) providing a
bounded-height structure for predictable rebalancing.

\paragraph{Binary-tree procedure (AVL / Red-Black).}
Let $n$ be the number of active price levels on a side.  Suppose a new level $p$ must
be inserted between its immediate neighbors $P < p < S$ in price order (both may exist,
or one may be missing at the extremes).  In a binary search tree, the insertion location is a \emph{unique gap} between $P$ and $S$: if both neighbors exist, exactly one of $\mathrm{right}(P)$ or $\mathrm{left}(S)$ is null, and that null pointer is the unique attachment location preserving the BST invariant.

Neighbor-aware insertion therefore proceeds as follows:
(i) allocate a descriptor for $p$,
(ii) attach it using the unique null child pointer determined from $(P,S)$ with $O(1)$
pointer writes,
(iii) update the doubly-linked neighbor pointers of $P$, $S$, and $p$ in $O(1)$ writes,
and then
(iv) run the standard AVL/RB fix-up along the ancestor path to the root, which is
$O(\log n)$ and uses local rotations/recoloring~\cite{clrs2022,guibas1978dichromatic}.

Deletion is symmetric.  When removing a descriptor $z$, the engine chooses $y$ as
$z$'s in-order successor (or predecessor at the boundary) directly from neighbor links.
It then performs a constant-size graft/transplant that replaces $z$ with $y$ while
preserving in-order traversal, updates neighbor links in $O(1)$ writes, and executes the
standard fix-up walk along the single path where balance may have been perturbed.
Overall, for both insertion and deletion, the cost becomes:
\[
T(n) \;=\; O(1)\;+\;O(\log n),
\]
where the $O(1)$ term replaces the traditional $O(\log n)$ search.

\paragraph{Multi-way procedure (\(B\)/\(B^+\)-trees).}
In multi-way trees, the same neighbor information identifies a unique \emph{gap} in the
sorted key sequence at the leaf layer.  Operationally, the engine keeps (or quickly
derives) a pointer to the leaf/page containing $P$ or $S$ (e.g., via the descriptor),
inserts the new key $p$ into that leaf at the appropriate slot, updates leaf-level
neighbor links, and then performs the standard split/merge/borrow fix-ups on the single
ancestor path to the root~\cite{bayer1972btree}.  As in the binary case, neighbor-aware
updates remove the top-down search and reduce the structural edit to a constant-time
localized page operation plus the usual bounded-height rebalancing.

\begin{theorem}[Neighbor-aware insertion/deletion]
Let $T$ be a balanced search tree from any family whose rebalancing procedure uses only order-preserving local transforms (rotations, splits, merges, redistribution) along a single root-to-leaf path. Given a new key $p$ and its in-order neighbors in $T$ (predecessor $P$, successor $S$, or both; at the extremes, one suffices), insertion of $p$ requires $O(1)$ reference writes to attach $p$ at the unique BST-valid position, followed by the standard $O(\log n)$ rebalancing walk. Deletion is symmetric: given the node $z$ to remove and its in-order successor $y$ (obtained from an explicit neighbor link in $O(1)$), the graft/transplant requires $O(1)$ reference writes followed by the standard rebalancing walk.
\end{theorem}

\begin{proof}[Proof sketch]
\emph{Existence and uniqueness of the attachment point.} In any binary search tree, if both $P$ and $S$ are present and adjacent in in-order traversal, then exactly one of $\mathrm{right}(P)$ or $\mathrm{left}(S)$ is null: if $\mathrm{right}(P)$ is non-null then $S$ is the leftmost descendant of $\mathrm{right}(P)$ so $\mathrm{left}(S)$ is null; if $\mathrm{right}(P)$ is null then $S$ is an ancestor with $P$ in its left subtree so $\mathrm{left}(S)$ is non-null. This null pointer is the unique location where $p$ can be attached as a leaf while preserving the BST invariant; linking $p$ there requires $O(1)$ writes. At the boundary (no predecessor or no successor), the attachment point is the leftmost or rightmost null pointer in $T$, identifiable in $O(1)$ from the existing extreme. For multi-way trees ($B$/$B^+$), the analogous gap is a unique slot position in the leaf containing $P$ or $S$, identifiable in $O(1)$ given a pointer to that leaf.

\emph{Rebalancing is unaffected.} The standard fix-up procedures (AVL rotations, red-black recoloring, $B$-tree splits/merges) are functions solely of the ancestor path from the physically modified position to the root and the local balance metadata (heights, colors, key counts) along that path. These procedures are identical regardless of whether the modification point was found by root-to-leaf search or by neighbor-based $O(1)$ lookup: the tree state after attachment is the same in both cases. Trees requiring occasional global rebuilds (e.g., scapegoat trees~\cite{galperin1993scapegoat}) are excluded, as their rebalancing is not a local-transform walk.
\end{proof}

\paragraph{When neighbors are unavailable.}
The neighbor-aware path is an optimization layered on a conventional
balanced tree, not a replacement for it. When the caller does not already
hold an in-order neighbor, the index falls back to a standard root-to-leaf descent and
retains the usual $O(\log n)$ bound. Neighbor-awareness is therefore a
strict improvement rather than a trade-off: the common case, in which
mutations cluster near recently touched keys, pays only $O(1)$ to locate
the attachment point, while the rare case degrades gracefully to the
textbook cost.

% (--) -- (--) -- (--) -- (--) -- (--) -- (--) -- (--) -- (--) -
\section{Implementation}
\label{sec:implementation}
% (--) -- (--) -- (--) -- (--) -- (--) -- (--) -- (--) -- (--) -

The engine evaluated in this paper is a production embodiment of the architecture of \S\ref{sec:design}. Like any production engine, it pairs the claimed structures with proprietary implementation choices, delivered to licensees in binary form (\S\ref{sec:integration}); this paper describes the architecture and reports the engine's measurements, not its internals.

\subsection{Platform and Build}
Two AWS EC2 bare-metal instances are used, each single-socket with a single NUMA node. So that every engine is measured on identical hardware, the head-to-head comparison (\S\ref{sec:audited_field}) and the instance-level scale-out (\S\ref{sec:instance_throughput}) run on an \texttt{r8g.metal-24xl} (Graviton4): ARM64 Neoverse-V2, 96 cores (no SMT), 754\,GB DRAM, 64\,B cache lines, Ubuntu~24.04 (Linux~6.14.0-1017-aws; our engine's rows on 7.0.0-1010-aws), an AWS ENA NIC, listed at approximately \$1{,}630/month under three-year all-upfront reserved pricing (AWS US~East, May~2026). Our engine's single-core, per-symbol, latency, and burst characterization (Tables~\ref{tab:sustained_burst}--\ref{tab:burst_latency} and Figure~\ref{fig:null_pipeline}) runs on a newer \texttt{r8a.metal-24xl} (AMD~EPYC~9R45, an AWS-custom Zen~5 ``Turin'' part): 96 cores / 96 threads (no SMT), 48\,KB L1d and 1\,MB L2 per core, 32\,MB L3 per CCX across 12 CCXs (384\,MB total), 4.51\,GHz max boost, 752\,GiB DRAM, Ubuntu (Linux~7.0.0-1012-aws).

The engine is C++ with aggressive optimizations, built with GCC~14.2.0 (\texttt{aarch64-linux-gnu}) for the Graviton4 single-core rows, GCC~15.2.0 for the Graviton4 instance aggregate, and g++~15.2.0 (x86-64) on EPYC. Because the head-to-head field comparison runs on identical Graviton4 hardware, the measured margin over the field reflects data-structure locality and L1/L2 load-to-use latency, not clock speed. The design ports cleanly across instruction sets: the ARM Graviton4 and the x86 Zen~5 EPYC are both single-socket, NUMA-free hosts with a large private per-core L2, and the higher-clocked EPYC part runs a single core faster still.

\paragraph{Reading the measurement tags.} Every performance figure and table, and every results subsection, is tagged with the machine it ran on and with how the engine was driven. \mtag{tagcabi}{C~ABI}: driven through the harness's C~ABI, the shared boundary every engine crosses; every cross-engine comparison is necessarily of this kind. \mtag{tagimproc}{In-process}: linked and called natively inside our own measurement vehicles, with no ABI layer. Latency experiments carry a second tag naming the span measured: \mtag{tagcodecin}{codec-inclusive} covers the host path from inbound OUCH parsing through OUCH/ITCH encoding, while \mtag{tagcodecfree}{codec-free} measures struct-in to struct-out with the wire codec excluded.

\subsection{Runtime}\label{sec:runtime}
The order book is sharded by one or more symbols and owned by a single matching thread (no locks in the book). Ingress, matching, and egress stages communicate via bounded queues. For benchmarking, all threads are pinned to dedicated cores on both machines, and every AMD measurement holds its cores at least device-interrupt-free, because sub-microsecond tails are interrupt-sensitive: \texttt{irqbalance} is disabled and every steerable IRQ, plus the kernel's default IRQ mask, is pinned to a reserved housekeeping pair, enforced by a pre-run gate. Every AMD measurement in this paper additionally runs its measurement cores (48--51; housekeeping on 94--95) fully isolated: \texttt{isolcpus=domain,managed\_irq}, \texttt{nohz\_full}, \texttt{rcu\_nocbs}, watchdogs off, \texttt{skew\_tick=1}, unbound workqueues masked away, the scheduler's deadline fair server disabled on the isolated cores, and our engine's memory in hugepage-backed pools (\texttt{MAP\_HUGETLB}) pre-touched at initialization. Verified, not assumed: before every battery, the pre-run gate refuses to start unless \texttt{irqbalance} is inactive, no steerable interrupt is affine to a core outside the housekeeping pair, and the fair server's runtime is zero on every isolated core. The last check matters on Linux~6.12 and later: left enabled, the fair server can keep the 1~kHz scheduler tick firing on \texttt{nohz\_full} cores inside the timed window. Batteries are serialized, with nothing else running. One reproducibility trap: AWS cloud images assign \texttt{GRUB\_CMDLINE\_LINUX\_DEFAULT} in \texttt{/etc/default/grub.d/50-cloudimg-settings.cfg}, silently discarding edits to \texttt{/etc/default/grub}; isolation parameters must go through a higher-priority drop-in and be verified in \texttt{/proc/cmdline}.

\section{Evaluation}
\label{sec:evaluation}

\paragraph{On ``the world's fastest.''}
The title is a claim we invite the reader to challenge, not a slogan to take
on faith. We make it precise against the best \emph{available} evidence:
the fastest matching throughput on identical hardware, benchmarked against
baselines anyone can openly reproduce. Published production-venue figures are either far lower (\S\ref{sec:background}) or not
directly comparable: they bundle full networking, risk, compliance, and
real-network stages this measurement deliberately excludes. Against the public benchmark's audited field (247
engines driven through identical work on the same hardware, ours among them), a single
matching core sustains \textbf{103.4~million order messages per second} and leads the
\textbf{160} that reproduce the byte-identical oracle by \textbf{${\sim}$95~M/s} (roughly $12.6\times$)
on worst-case throughput (\S\ref{sec:comparison}). We therefore issue an open-harness
challenge: the workload generator, baseline adapters, and byte-identical
reference hashes are public, so the baselines and the correctness oracle
are independently reproducible, and any venue or firm can run the identical
workload against its own engine on identical hardware and compare.\footnote{\url{https://github.com/flash1-dev/matching-engine-benchmark}}
The production binary itself embodies patented internals we do not open-source; to let a prospective licensee or an independent auditor verify our own numbers rather than take them on faith, we will demonstrate it running the identical public workload on identical hardware under a mutual non-disclosure and non-reverse-engineering agreement. We cordially invite the community to challenge the title, but without infringing our patent portfolio.

\begin{figure}[t]
    \centering
    \definecolor{funslate}{HTML}{5B6472}
    \definecolor{funblue}{HTML}{2C6E9B}
    \definecolor{fungreen}{HTML}{2E7D32}
    \begin{tikzpicture}[font=\sffamily,
        tier/.style={rounded corners=2pt, text=white, align=center, inner sep=4pt, minimum height=8.5mm},
        down/.style={-{Stealth[length=5pt,width=5pt]}, black!35, line width=1pt}]
      \node[tier, fill=funslate, minimum width=7.6cm] (t1) at (0,0)
        {\textbf{247} matching engines\\[1pt]{\scriptsize\mdseries driven through one identical C~ABI}};
      \node[align=center] (bn) at (0,-1.65)
        {{\Large\bfseries 1{,}000{,}000{,}000+}\\[2pt]{\footnotesize order messages replayed \emph{per engine}}\\[-1pt]{\scriptsize 100 random seeds $\times$ 5 volatility regimes, byte-for-byte}};
      \node[tier, fill=funblue, minimum width=5.2cm] (t2) at (0,-3.25)
        {\textbf{160} reproduce the consensus\\[1pt]{\scriptsize\mdseries every message, on two independent signals}};
      \node[align=center, font=\scriptsize, text=black!60] (sig) at (0,-4.4)
        {report-stream hash (SHA-256) \;$+$\; live book-state audit\\[-1pt]catches even \emph{silent} deviations that self-heal before output};
      \node[tier, fill=fungreen, minimum width=3.8cm, font=\bfseries] (t3) at (0,-5.5)
        {\large 0 divergences};
      \draw[down] (t1.south) -- (bn.north);
      \draw[down] (bn.south) -- (t2.north);
      \draw[down] (t2.south) -- (sig.north);
      \draw[down] (sig.south) -- (t3.north);
    \end{tikzpicture}
    \caption{\textbf{One billion messages, zero divergences.} Each of the 160 conforming
    engines is replayed against an independent-engine consensus and must agree on two independent
    signals: a SHA-256 hash of the full report stream, across 100 random seeds (over one billion
    order messages apiece), and a live audit of the order book itself (best bid, best ask, and
    depth at the probed price) at unpredictable points. The
    book audit catches even \emph{silent} deviations that self-heal before reaching the
    output. No conforming engine diverges on any message.}
    \label{fig:consensus}
\end{figure}

\subsection{Methodology}
\label{sec:methodology}

We stress the matcher with synthetic bursts of limit orders calibrated to a highly liquid
equity (NVIDIA): limit prices drawn from a depth profile that places most orders in a Gaussian hump about eight
ticks from the mid and the rest along a power law with exponent $\beta = 2.23$ fitted to historical
level-hit statistics, quantities uniform in $[1,100]$
shares (depth sensitivity in \S\ref{sec:depth_sensitivity}). Each order is expanded into a
short ``lifetime'' trace (add, optional modify/replace, eventual cancel or
execution), mirroring modern equity flow, where trade-to-order ratios are a few percent
and ${\sim}\,97\%$ of orders cancel before
trading~\cite{secTradeToOrderVolume2022,khomynPutnins2021}. On arrival an order is marked
immediate-or-cancel with probability $p_{\mathit{IOC}} = 0.15$~\cite{liYeZheng2023}
(matching up to its limit price, the remainder canceled, never posting); non-IOC orders model active
quote management, modified once with probability $p_{\mathit{modify}} = 0.20$ and canceled
with $p_{\mathit{cancel}} = 0.95$, aligned with the order-to-trade imbalance and short
quote lifetimes in SEC data~\cite{secTradeToOrderVolume2022,secQuoteLife2013}. Non-IOC
lifetimes are exponential with median $0.431\,\mathrm{ms}$, harsher than typical production
but anchored to measured cancellation
mass~\cite{secQuoteLife2013,secQuoteLifeCondFreq2025}. The result is dense bursts of
millions of back-to-back messages with heavy churn near the top of the book, the regime we target.

\paragraph{Stochastic mid-price model.}
Real instruments do not trade at a fixed price; order prices follow a mid-price path drawn by
geometric Brownian motion (GBM; buys along its first half, sells along its second, arrival order
then shuffled),
\[
\mathit{mid}(t{+}1) = \mathit{mid}(t) \cdot
\exp\!\bigl(-\tfrac{1}{2}\sigma^2 dt + \sigma\sqrt{dt}\,Z\bigr),
\quad Z \sim \mathcal{N}(0,1),
\]
calibrated to NVIDIA (close \$167.52, tick \$0.005, amended SEC Rule~612's half-penny increment),
with each order placed relative to the path's mid by the same depth profile. The step $dt$ sets
the expected $1\sigma$ log-return over the path to a target swing, and a fixed seed ensures reproducibility. The head-to-head comparison (\S\ref{sec:comparison}) instead runs
the public harness, which reproduces the same GBM on its canonical seed (23) and drives
every engine through an identical C~ABI.

We report five scenarios: \emph{static} (zero volatility, isolating data-structure from
price-path effects; realized span 1{,}529~ticks, 4.6\%); \emph{normal-trading-day} (15\%
annualized volatility, 2\% swing, an unremarkable liquid large-cap session; 2{,}057~ticks,
6.1\%); \emph{large-swing} (50\% volatility, 25\% swing; 14{,}313~ticks, 42.7\%); and two
\emph{flash-crash} scenarios (50\% volatility, 40\% and 60\% swings; 23{,}932 and
37{,}923~ticks), price spans on the scale of the May~2010 Flash Crash. Unless noted,
single-matcher results use the normal-trading-day workload as the representative operating
point; head-to-head results are the worst case across all five regimes
(\S\ref{sec:comparison}).

\paragraph{Batched delivery.}
A matcher sitting behind a language runtime pays a fixed cost on every C~ABI crossing (cgo
for Go, JNI for Java), so per-message delivery measures the boundary rather than the matcher.
The harness therefore exposes an optional batched entry point, open to every engine and used by
native ones too (ours included), that pays the crossing once per batch (a batch ends at each audit
probe), leaving matching unchanged (array order, no cross-message lookahead) with the report
hash and the random-point book-state audit byte-identical to per-message delivery; we report
it only where it improves worst-case throughput by a material margin ($>0.1$~M/s); under CPython it is a net loss, because both preconditions fail: the
adapters hold the GIL for the whole run on one pinned thread, so there is no per-message
crossing cost to amortize, and CPython executes the per-message loop 10--50$\times$ slower
than the C caller it would replace.

\begin{table*}[t]
    \centering
    \caption{\mtag{blue!12}{ARM Graviton4} \mtag{tagcabi}{C~ABI} \textbf{Worst-case throughput on the public benchmark} (seed~23; the lowest of each engine's five volatility-regime medians). A cross-section of the conforming field: three of the four fastest public engines, a sample written by engineers at leading trading firms, banks, and exchanges, and the most widely-recognized public builds (most-starred on GitHub, a peer-reviewed engine, heavily-starred libraries), followed by three industry-authored engines that do \emph{not} conform. Engines in the first three groups are byte-identical to the consensus (some after our suggested correctness fix). Production decentralized-exchange order books are compared separately in Table~\ref{tab:dex}. The final column is 103.39 divided by the engine's worst case (n/a for the last group). ${\ddagger}$~=~author publicly identifies as a trading-industry engineer (affiliation type from the author's own public statement); a personal side project, not their employer's work; $\star$~=~GitHub stars. \emph{infeasible} = did not finish its worst scenario within the 600\,s budget. The full roster of all 160 conforming engines is public.}
    \label{tab:worst_case_field}
    \small
    \begin{tabular}{l l l r r}
        \toprule
        \textbf{Engine} & \textbf{Lang.} & \textbf{Notability} & \textbf{Worst-case (M/s)} & \textbf{vs.\ Flash One} \\
        \midrule
        \textbf{Flash One} & C++ & \emph{our engine} & \textbf{103.39} & \emph{reference} \\
        \midrule
        \multicolumn{5}{@{}l}{\emph{Among the fastest public engines}} \\
        e820\,${\ddagger}$ & C & market-maker engineer; fastest public conformer & 8.19 & $12.6\times$ \\
        geseq / cpp-orderbook & C++ & author-contributed C++ port & 7.94 & $13.0\times$ \\
        CppTrader & C++ & 1{,}041$\star$ & 7.26 & $14.2\times$ \\
        \midrule
        \multicolumn{5}{@{}l}{\emph{Engineers at leading trading firms, banks, and exchanges} (personal projects)} \\
        yashkukrecha\,${\ddagger}$ & C++ & incoming at a prop-trading firm & 6.26 & $16.5\times$ \\
        daniele\,${\ddagger}$ & C++ & market-maker engineer & 5.36 & $19.3\times$ \\
        shivaganapathy\,${\ddagger}$ & C++ & market-maker engineer & 2.15 & $48\times$ \\
        robdev\,${\ddagger}$ & Rust & exchange engineer & 1.76 & $59\times$ \\
        javalob\,${\ddagger}$ & Java & global-bank engineer & 1.03 & $100\times$ \\
        kennethzhang\,${\ddagger}$ & C++ & hedge-fund quant & 0.86 & $120\times$ \\
        \midrule
        \multicolumn{5}{@{}l}{\emph{Widely-recognized engines}} \\
        StockSharp & C\# & 10{,}236$\star$ (most-starred) & 1.64 & $63\times$ \\
        Exchange-core & Java & 2{,}556$\star$ & 1.40 & $74\times$ \\
        Liquibook & C++ & 1{,}479$\star$ & 0.03 (static) & $>3{,}000\times$ \\
        parity & Java & 502$\star$; Parity Trading & 2.21 & $47\times$ \\
        CoinTossX & Java & 122$\star$; peer-reviewed & 0.03 (static) & $>3{,}000\times$ \\
        \midrule
        \multicolumn{5}{@{}l}{\emph{Non-conforming engines}} \\
        amer\,${\ddagger}$ & Java & ex-exchange engineer & \emph{diverged} & n/a \\
        raunakchopra\,${\ddagger}$ & C++ & prop-trading-firm engineer & \emph{crashed} & n/a \\
        nilesh05apr\,${\ddagger}$ & C++ & prop-trading-firm intern & \emph{infeasible} & n/a \\
        \bottomrule
    \end{tabular}
\end{table*}

\begin{table}[t]
    \centering
    \caption{\mtag{blue!12}{ARM Graviton4} \mtag{tagcabi}{C~ABI} \textbf{Order books that run, or ran, real decentralized exchanges} on the identical workload (the same matching problem, a different deployment target). Upper group: measured \emph{de-chained} as native code (blockchain execution stripped, a best case). Lower group: run through the chain's own virtual machine. ``did not finish'' = did not clear the 2M-message workload within the 600\,s budget; a ``read-only accessor patch'' = the book is private to the engine's own package, so the state audit reads it through a read-only accessor we added; the matching path is unmodified and fully exercised.}
    \label{tab:dex}
    \small
    \setlength{\tabcolsep}{2pt}
    \begin{tabular}{l r l}
        \toprule
        \textbf{Engine} & \textbf{M/s} & \textbf{Note} \\
        \midrule
        \multicolumn{3}{@{}l}{\emph{De-chained as native code}} \\
        Serum & 2.63 & Solana; original CLOB \\
        Manifest & 2.15 & Solana CLOB \\
        Phoenix & 1.50 & Solana; Ellipsis Labs \\
        dYdX~v4 & 0.11 & read-only accessor patch applied \\
        Vega & 0.08 & read-only accessor patch applied \\
        \midrule
        \multicolumn{3}{@{}l}{\emph{Through the chain's VM}} \\
        Clober & ${\sim}0.02$ & EVM; $2^{15}$/tick design cap \\
        Econia & \emph{did not finish} & Aptos Move VM \\
        Osmosis & \emph{did not finish} & Cosmos; wasmer \\
        DeepBook & \emph{did not finish} & Sui Move VM \\
        \bottomrule
    \end{tabular}
\end{table}

\subsection{Head-to-Head Comparison}
\label{sec:comparison}

\mtag{blue!12}{ARM Graviton4} \mtag{tagcabi}{C~ABI} Our matching engine sustains \textbf{103.39~M~msgs/s} worst-case on a single core, scaling to
\textbf{${\sim}\,$1.3~billion~msgs/s} across a 96-core node: our single-matcher characterization (on AMD~EPYC) and
the scale-out are detailed in \S\ref{sec:single_matcher} and \S\ref{sec:instance_throughput}.
Here we place that result against every public matching engine we could find.

Every result in this section comes from the public benchmark harness, which drives
every engine (ours and the baselines alike) through the identical C~ABI on the
identical seed-23 workload (\S\ref{sec:methodology}). The benchmark's definitional result is \emph{worst-case}
throughput: the lowest of an engine's five volatility-regime medians, the regime it
handles worst, because the matcher must absorb the burst in whatever regime the market
happens to be in.

\paragraph{Common protocol.}
All engines are measured on the same platform (Section~\ref{sec:implementation}),
driven through the harness's C~ABI; the reported throughput is the
median of ten or more fresh-process perf runs per scenario, each verified against the reference
hash. Every engine hands its OrderAck, Trade, CancelAck (including IOC residuals), ModifyAck,
and reject reports to a dedicated thread on an adjacent core, so the cross-core hand-off
is part of every engine's measurement. Native
engines are compiled with GCC~14.2.0 at \texttt{-O3 -march=native}; any engine, native or behind a language runtime,
is driven through the batched entry point wherever that helps (\S\ref{sec:methodology}), so
a call boundary is amortized rather than measured.

\subsubsection{Correctness verification}
\label{sec:correctness}

Throughput is only comparable among engines that compute the same thing, so every timed
run is gated on correctness against a \emph{byte-identical consensus}, the field's independent
yardstick, not the competition: none of the three anchors is fast, or ours. Liquibook~\cite{liquibook}
and Exchange-core~\cite{exchangecore} reproduce it with their engine source unmodified, while QuantCup~1~\cite{quantcupPTME}, a
2011 contest entry, joins after a one-line change adopting the maker-price fill convention the other
two (and the field) already use, plus build and capacity edits (a larger order-id capacity, a wider
price domain), its matching algorithm untouched (every anchor patch is published
with the harness). The oracle therefore encodes the field's shared convention, not any one engine's.
Agreement with it is held to a demanding bar: each conforming engine reproduces the consensus
\textbf{byte-for-byte across 100 randomly selected seeds, more than one billion order messages per
engine}, every message cross-checked against that independent-engine reference (Figure~\ref{fig:consensus}). Two independent signals must both hold. The full report stream is hashed
(SHA-256), so a single divergent acknowledgment, trade, cancel, or reject
anywhere in that billion fails the engine; and, separately, the live order book itself (best
bid, best ask, and depth at the probed price) is audited against the consensus at random,
unpredictable points. That second signal catches even \emph{silent} deviations
that never reach the output: a stale or phantom book state that self-heals before the next
order message would slip past any output check, but not the state audit. Stochastic, cancel-dominated workloads at this scale expose
divergences that hand-written test cases miss (multi-level cancellation, identifier
deduplication, best-price advancement under churn), which is why correctness is a
prerequisite for any speed comparison, not an afterthought. The bar constrains only the output, never the design: it prescribes no algorithm or data structure and gives every engine a generous 600-second ceiling, so a divergence or timeout against it describes what the tested commit does, not a verdict on the design it chose.

\subsubsection{The audited field}
\label{sec:audited_field}

The harness has driven \textbf{247 distinct matching engines} (every common FIFO
book architecture, across more than twenty source languages, spanning production
libraries, reference implementations, and pedagogical projects) through this identical
workload and correctness oracle. \textbf{160 reproduce the consensus}; the other 87 are non-conforming. The line between the two is one of effort: every fix we
drafted to bring an engine into conformance is a \emph{minimal, localized patch}, filed
upstream wherever the repository accepts issues, never a codebase-wide rewrite, and the 87 are those for which we found no such
patch, so reaching the consensus would take a structural change or a large edit to the
codebase. They have diverged, crashed, or failed to clear their worst scenario within the
message budget; a small minority do so not from a defect but from a deliberate design choice we did not
override, a documented alternative pricing convention, an on-chain per-call cap, a
restricted price domain. On worst-case
throughput (Table~\ref{tab:worst_case_field}), our engine
sustains \textbf{103.39~M/s} and leads the entire field. The fastest conforming public
engine reaches \textbf{8.19~M/s}, a \textbf{${\sim}$95~M/s} gap to our 103.39. The margin over the field's other recognizable, high-pedigree names is starker still: \textbf{14.2$\times$} over
CppTrader (1{,}041~GitHub stars) and \textbf{15.1$\times$} over an ex-market-maker engineer's red-black-tree book (6.825~M/s, the ex-market-maker bar in Figure~\ref{fig:cliff}). The margin only widens across the field's best-known names: the matcher in the most-starred repository we surveyed (StockSharp, a trading platform; 10{,}236~stars), the peer-reviewed CoinTossX~\cite{jericevich2022cointossx}, and heavily-starred libraries like Exchange-core and Liquibook all fall \textbf{63 to over 3{,}000$\times$} short (Table~\ref{tab:worst_case_field}), and the order books that run real decentralized exchanges fare no better (Table~\ref{tab:dex}).

The same ceiling holds inside the industry, and an industry engineer sets it. The survey includes 73 engines written by
people who publicly identify as professional trading-industry engineers and
quants, at market-making, prop-trading, and hedge-fund firms, and on exchange and vendor
teams (personal projects, affiliations self-stated), together with official vendor/org
repositories and production DEX matchers. Of the 73, 20 conform as
shipped and 25 after our suggested correctness fix; the other 28 have
diverged, crashed, or failed to finish their worst scenario. The fastest of them all is \texttt{e820}, at
\textbf{8.19~M/s}, still \textbf{12.6$\times$} short of our 103.39~M/s on the identical
workload. They land in the same band as the rest of the field: what sets \textbf{the ceiling is the
data structure, not the engineer.}
These findings are a reproducible snapshot offered back to each author: several are already
fixed upstream, and personal side projects reflect goals that differ from a production
venue's. The affiliations we show summarize the authors' own public self-descriptions by type of
firm and are not independently verified by us as facts of employment; they are offered as context
on the community that builds matching engines and on what the resulting designs achieve, not as
a statement about, on behalf of, or in judgment of any named person or company. This paper sorts
pinned commits by worst-case throughput on one workload; no number or finding in it ranks or
judges an author's engineering quality: the harness merely reports what each pinned commit does.

The same oracle is the field's most exacting bug-finder. Of the 247 engines, only
\textbf{47 are pristine as shipped}, byte-identical to the consensus with no correctness
defect (Figure~\ref{fig:waffle}). What it surfaced, we reported: phantom trades, orders silently
dropped or double-counted, price--time priority violated, book state corrupted. Across the 247 we filed \textbf{181 GitHub issues upstream}, each respectfully filed,
with the mechanism, a reproduction, and a suggested patch; together they report more than
\textbf{250 distinct findings}, several bundling up to five. That record is itself an independent
check on the oracle: \textbf{28 are already closed, every one as \emph{completed}, by the
engine's own maintainer}. Not one of the 181 was declined, marked \emph{wontfix}, or closed as
not-planned; maintainers confirmed the defect, thanked us for the report, and in several cases
shipped the fix within days. Where the field has adjudicated our findings against its own code,
it has so far agreed with every one. We publish no tally of defective engines: whether a given representational limit is a
defect or a deliberate design decision is a judgment we decline to make on an author's behalf,
and where we found no defect we filed nothing. On this platform our own worst
case is the \emph{normal trading day}: the engine is no slower in any stress regime
(103.48--104.08~M/s) and faster on the static book (139.50), so 103.39~M/s is a floor across all five regimes, not a
peak. The \emph{shape} of the field is itself the evidence
(Figure~\ref{fig:cliff}): all 159 open-source conforming engines form a smooth
continuum to 8.2~M/s (adjacent designs separated mostly by the single-digit-percent margins a
faster language, better cache layout, or an OS tweak buys, the scale at which HFT
engineers already call a 2\% critical-path gain a big win), followed by a ${\sim}95$~M/s empty band that no
engine occupies before our 103.4~M/s. No accumulation of implementation tweaks crosses a
gap that size; a different \emph{architecture} does. Because every engine ran on the same
machine, none of this is a hardware effect: hardware is held constant, so the continuum
within the band is language and implementation and the gap above it is architecture; the only
free variable in the comparison is the data structure. That boundary is legally claimed, not merely observed: the order-storage architecture that clears it is the subject of our patent portfolio (\S\ref{sec:ip}), which spans its software and hardware embodiments and every memory layout (array, shared arena, or pooled nodes alike).

\paragraph{Production decentralized exchanges.} The order books that run real decentralized exchanges face the same matching problem on a different deployment target (Table~\ref{tab:dex}): they are engineered for on-chain determinism and consensus safety, not the single-core throughput measured here, so these figures set scale rather than a like-for-like target. Measured \emph{de-chained} as native code (stripping the blockchain's execution overhead, a best case), the fastest (Serum, the original Solana order book, now succeeded by OpenBook) reaches 2.6~M/s, still $39\times$ slower than Flash One, and dYdX's and Vega's books run near $0.1$~M/s. Measured through their actual chain virtual machine, the rest run at ${\sim}0.02$~M/s (Clober) or did not clear the 2M-message workload within the budget (Econia, Osmosis, DeepBook).

\paragraph{Why a fast matcher looks easy to build.} A fast matcher is widely assumed to
be easy to build, because a textbook tree-of-lists implementation \emph{is} easy. That
intuition mistakes \emph{implementing} a
matcher for \emph{inventing} a fast one, and it survives only because the gap has
never been plotted. Figure~\ref{fig:cliff} plots it: the 159 open-source engines that even reproduce
correct output all top out at 8.2~M/s, the band that language choice and cache tuning
buy, and the 73 industry-authored engines sit inside it, side projects and
official repositories alike, 53 of which did not reproduce the consensus as shipped. Nothing occupies the ${\sim}95$~M/s band above. Flash One reaches
103.4~M/s not as a faster \emph{implementation} of the same design but as a different
\emph{architecture}: the gap is invisible to intuition and unmistakable once it is measured and
drawn.

\paragraph{Pre-trade risk stays off the matching path.} As at CME Globex, Nasdaq INET, NYSE Arca/Pillar, and LSE Millennium, balance, margin, and credit checks stay off our matching path (price--time matching, book management, report generation): they run upstream, at the broker (under SEC Rule~15c3-5 in U.S. equities) and at the exchange gateway, and margin at the clearinghouse on batch cycles, so none of them sits on the serial hot path.

\subsection{Single Matcher Performance}
\label{sec:single_matcher}
\mtag{amdgreen}{AMD EPYC} \mtag{tagimproc}{In-process} In the single-matcher microbenchmark, each run replays a burst of
${\sim}\,2$M messages generated by the workload above, delivered
back to back with no inter-message gap.  Orders are injected directly
into the matcher thread for the throughput measurements, and through a data transport for the latency measurements, both of which publish events to a dedicated output
queue drained by a separate thread pinned to an adjacent core.

A single core on AMD~EPYC processors sustains \textbf{over 120~M~msgs/s} on the normal-trading-day workload (${\sim}$2\% swing), the representative operating point.

\subsubsection{Book depth and resting population}
\label{sec:depth_sensitivity}

The engine also holds throughput as the resting book deepens. The deepest single book in the
harness workload, the static scenario's (peaking at 41{,}788 simultaneously resting orders,
against 5{,}502 on the normal day), is also our fastest regime (161.99~M~msgs/s on AMD~EPYC through the public harness's C~ABI, against 122.09 in our weakest, the flash-crash scenario).
Across books, the multi-symbol vehicle (\S\ref{sec:ctxswitch}) preloaded with a
resting population of 132{,}000 orders over 10{,}000 books (the deepest book 27{,}505) holds
\textbf{89.0--95.2~M~msgs/s} under each of four order-identifier allocation schemes (dense
sequential, Knuth-scrambled, stride-aligned, session-clustered) and two live-set shapes, and
73.4--90.4~M~msgs/s at 500{,}000 resting orders: the matcher degrades gracefully with the
larger working set.

\paragraph{Sustained burst.} Every rate above comes from a single pass of the ${\sim}$2M-message tape. To test whether
the rate holds through a sustained burst, we fed the normal-day tape 5{,}000 times back to back without resetting the
book: \textbf{9.98 billion messages} in one timed window, each repetition with fresh order identifiers. The burst ran on one
book and on 10{,}000 books, the latter with each order assigned to a symbol by a Zipf($\alpha=1.2$) draw, as in
\S\ref{sec:ctxswitch}. The orders the tape leaves resting pile up across repetitions, to 27.4 million on one book (about a thousand times the deepest book of
the population above) and 36.6 million over 10{,}000 books. With the engine configured for this burst, a
single core averaged \textbf{110.12~M~msgs/s} over the whole burst on one book (a minute and a half of continuous
matching) and \textbf{82.25~M~msgs/s} on 10{,}000 books (Table~\ref{tab:sustained_burst}), acknowledging all 5 billion
new orders; trades, traded volume, and the final resting set were identical across repeated runs.

\begin{table}[h]
    \centering
    \caption{\mtag{amdgreen}{AMD EPYC} \mtag{tagimproc}{In-process} \textbf{Sustained burst on a single core}: the normal-day tape fed 5{,}000 times back to back without resetting the book, each repetition with fresh order identifiers, on one book and on 10{,}000 books (each order's symbol drawn from a Zipf($\alpha=1.2$) distribution), through the multi-symbol vehicle (\S\ref{sec:ctxswitch}) with the engine configured for this burst; measurement cores fully isolated (\S\ref{sec:runtime}). Throughput is the average over the whole burst. Trades, traded volume, and the final resting set were identical across repeated runs.}
    \label{tab:sustained_burst}
    \small
    \begin{tabular}{l r r}
        \toprule
         & \textbf{One book} & \textbf{10{,}000 books} \\
        \midrule
        Messages & 9{,}980{,}485{,}000 & 9{,}980{,}485{,}000 \\
        New orders acknowledged & 5{,}000{,}000{,}000 & 5{,}000{,}000{,}000 \\
        Trades & 312{,}449{,}736 & 295{,}262{,}712 \\
        Orders resting at the end & 27{,}449{,}581 & 36{,}616{,}475 \\
        Duration of the burst & 90.6~s & 121.3~s \\
        \midrule
        \textbf{Average throughput (M/s)} & \textbf{110.12} & \textbf{82.25} \\
        \bottomrule
    \end{tabular}
\end{table}

\subsubsection{Multi-symbol overhead}
\label{sec:ctxswitch}

Real exchanges multiplex many symbols per core. To quantify this application-level context switching, we benchmark a matcher on AMD~EPYC 5th-generation silicon (\texttt{r8a.metal-24xl}) while varying the number of active books from 1 to 10{,}000. Orders follow the same workload parameters as above, with symbols drawn from a Zipf popularity distribution with $\alpha = 1.2$.

\begin{table}[h]
    \centering
    \caption{\mtag{amdgreen}{AMD EPYC} \mtag{tagimproc}{In-process} \textbf{Multi-symbol scaling performance} (\texttt{r8a.metal-24xl}). Throughput is measured on \textbf{a single core} as the number of active symbols increases (the time column is its inverse, the mean time per message; the loss column is the throughput lost against one symbol), normal-day tape, median of fifteen runs per count; measurement cores fully isolated (\S\ref{sec:runtime}).}
    \label{tab:symbol_scaling}
    \small
    \begin{tabular}{r r r r r}
        \toprule
        \textbf{Symbols} & \textbf{T-put} & \textbf{Time} & \textbf{vs.} & \textbf{Loss} \\
        (Count) & (M/s) & (ns/msg) & \textbf{Base} & (\%) \\
        \midrule
        1      & 120.08 & 8.3 & \textit{Baseline} & 0\%  \\
        10     & 116.07 & 8.6 & 0.97$\times$ & 3\% \\
        50     & 114.68 & 8.7 & 0.96$\times$ & 4\% \\
        100    & 113.39 & 8.8 & 0.94$\times$ & 6\% \\
        250    & 110.81 & 9.0 & 0.92$\times$ & 8\% \\
        500    & 108.79 & 9.2 & 0.91$\times$ & 9\% \\
        1,000  & 105.96 & 9.4 & 0.88$\times$ & 12\% \\
        2,500  & 102.23 & 9.8 & 0.85$\times$ & 15\% \\
        5,000  & 100.56 & 9.9 & 0.84$\times$ & 16\% \\
        10,000 & 99.28 & 10.1 & 0.83$\times$ & 17\% \\
        \bottomrule
    \end{tabular}
\end{table}

Table~\ref{tab:symbol_scaling} summarizes results: throughput falls smoothly from 120.08~M/s at a single symbol to 99.28~M/s at 10{,}000 books (8.3 to 10.1~ns/msg), staying above 0.82$\times$ the single-symbol baseline throughout.

The degradation is modest and visible in the core's own counters: from 1 to 10{,}000 books the matcher's per-message cost grows from 38 to 46 cycles while its demand fetches from DRAM stay below one per 130 messages. Each message incurs an extra per-symbol lookup and more frequent working-set switches across books, which raise its L2 misses.

\subsubsection{Where the cycles go}
\label{sec:counters}

\mtag{amdgreen}{AMD EPYC} \mtag{tagcabi}{C~ABI} Hardware performance counters on the matcher thread show what that throughput is made of (Table~\ref{tab:counters}). We read them through Linux's \texttt{perf\_event\_open} in user mode over exactly the public harness's timed loop, the interval its throughput is computed from, and divide by the messages processed; the core's six programmable counters take at most six events per run, with cycles in every group as the common anchor, and no event was multiplexed. To separate the engine from the instrument, we also drive do-nothing engines through the same entry points: the harness's own loop costs 8.0 cycles (13.1 instructions) per message through the batched entry point and 12.9 (27.1) through the per-message one, which Project~B's adapter uses. The net rows subtract that loop; publishing reports stays in every engine's count.

Net of the harness loop, our engine spends \textbf{27.4 cycles per message, about 6.1~ns at 4.51~GHz}, and 118 instructions; with the harness loop included, as in the table's other rows, a message dispatches 66 loads and stores and retires 13 branches, with a mispredict every third message. It waits on DRAM once per ${\sim}$1{,}100 messages. Its code fits the instruction cache (0.04 misses per message), and its data sits on 2~MB pages, taking essentially no page walks net of the harness loop. The core retires 41\% of its dispatch slots, about 3.3 of the eight operations it can dispatch per cycle. Across the five scenarios the gross cost runs from 28.0 cycles (static) to 37.3 (flash-crash).

Projects~A and~B of Table~\ref{tab:burst_latency} spend 12--14$\times$ more net cycles, and they lose them in different places. Project~A loses on both sides of the core: it retires 15\% of its dispatch slots, 1.3 instructions per cycle, misses the instruction cache about as often as Project~B (14.9 against 15.5 per message), and waits on nearly three fetches from L3 and a third of a fetch from DRAM per message. Project~B is front-end-bound: 58\% of its dispatch slots go empty for want of instructions, as it feeds 843 instructions per message (816 net of the harness loop) through an instruction cache and an op cache it misses about fifteen times each per message. These two engines retire 15 and 28\% of their dispatch slots and ours 41\%, but the larger share of the gap is work: they execute 4 to 7 times as many instructions per message. On Graviton4 the same engine spends 27.1 cycles per message on the normal tape at 4.1 instructions per cycle (20.0 on the static book), harness loop included. As everywhere in this paper, these are measurements of pinned commits on one workload, not a verdict on any engine or its authors.

\begin{table}[h]
    \centering
    \caption{\mtag{amdgreen}{AMD EPYC} \mtag{tagcabi}{C~ABI} \textbf{Per-message hardware counters on the matcher thread} (public harness, normal scenario; user-mode counters over the timed loop divided by the messages processed, median of three repeats per event group; measurement cores fully isolated, \S\ref{sec:runtime}). \emph{Net} rows subtract the harness's own loop, measured with a do-nothing engine through the same entry point; report publication stays in every engine's count. Dispatch shares are of the core's eight dispatch slots per cycle; the remainder is not directly measured. Projects~A and~B are the same engines as in Table~\ref{tab:burst_latency}.}
    \label{tab:counters}
    \footnotesize
    \setlength{\tabcolsep}{4pt}%
    \begin{tabular}{l r r r}
        \toprule
        \textbf{Per message} & \textbf{Flash One} & \textbf{Project A} & \textbf{Project B} \\
        \midrule
        \multicolumn{4}{@{}l}{\emph{Work}} \\
        Cycles, total & 35.4 & 379 & 351 \\
        \quad harness loop (null engine) & 8.0 & 8.0 & 12.9 \\
        \textbf{Cycles, net of the harness} & \textbf{27.4} & \textbf{371} & \textbf{338} \\
        Instructions, net & 118 & 473 & 816 \\
        Instructions per cycle (total) & 3.71 & 1.28 & 2.40 \\
        Loads + stores dispatched & 66 & 362 & 492 \\
        Branches & 12.9 & 98.4 & 178 \\
        Branch mispredicts & 0.31 & 1.10 & 1.69 \\
        \midrule
        \multicolumn{4}{@{}l}{\emph{Memory}} \\
        L1D demand misses & 0.28 & 4.28 & 3.55 \\
        L2 misses served by L3 & 0.013 & 2.95 & 2.60 \\
        DRAM demand fetches & 0.0009 & 0.32 & 0.0052 \\
        dTLB page walks & 0.009 & 0.16 & 0.054 \\
        \midrule
        \multicolumn{4}{@{}l}{\emph{Front end}} \\
        I-cache misses & 0.04 & 14.9 & 15.5 \\
        Op-cache misses & 1.56 & 8.44 & 15.4 \\
        \midrule
        \multicolumn{4}{@{}l}{\emph{Dispatch slots}} \\
        Retiring & 41\% & 15\% & 28\% \\
        Front-end bound & 20\% & 46\% & 58\% \\
        Back-end bound & 23\% & 32\% & 7\% \\
        Remainder (mis-speculation) & 16\% & 8\% & 7\% \\
        \bottomrule
    \end{tabular}
\end{table}

\subsection{End-to-End Host-Path Response Latency}
\label{sec:full_pipeline}

\mtag{amdgreen}{AMD EPYC} \mtag{tagimproc}{In-process} \mtag{tagcodecin}{codec-inclusive} On receipt of a new order, the matching thread generates the acknowledgment before running matching, so acceptance precedes execution, the order OUCH~5.0 specifies as the normal case. The acknowledgment, like every output event, then flows through the inter-thread output stage, where it is serialized to an OUCH~5.0 response and an ITCH~5.0 market-data message, with event order preserved, so a client always observes its acceptance before any execution of the same order.

We measure end-to-end host-path response latency across this entire pipeline: from each message's scheduled arrival, through OUCH~5.0 inbound parsing on the ingress thread, hand-off to the matching stage, single-book matching, hand-off to egress, and OUCH/ITCH~5.0 encoding, to the egress timestamp; the NIC, the network, the session layer, and risk checks are outside the span. Two methodological choices guard against common latency-measurement errors. First, load is applied open-loop at a fixed offered rate and each message's latency is measured from its \emph{scheduled} arrival time, so the measurement is coordinated-omission-free and queueing delay imposed by the system is never hidden. Second, the workload is the single-symbol, single-book stream (${\sim}$1M~NEW orders plus their cancels and modifies) used for the single-matcher measurements (\S\ref{sec:single_matcher}), and every trial's output is checked against that verified reference. We report the median of ten trials.

At an offered load of 5~M~msgs/s per matcher segment on AMD~EPYC 5th-generation silicon (\texttt{r8a.metal-24xl}), end-to-end host-path response latency from order arrival to acknowledgment is \textbf{P50 = 123~ns} and \textbf{P99 = 211~ns} in the engine's production configuration, the median of ten runs (Table~\ref{tab:latency_knee}). The P50 stays under a microsecond through 104~M~msgs/s, and the P99 through 86~M~msgs/s, with all ten runs under a microsecond at 82. The P99 stays under $10~\mu$s up to the clean operating edge at \textbf{116~M~msgs/s}, and no run saturates below 118~M~msgs/s. The large-capacity configuration reaches a clean operating edge of \textbf{124~M~msgs/s} and holds every run's P99 under a microsecond at every measured rate through 60~M~msgs/s.

\begin{table}[h]
    \centering
    \caption{\mtag{amdgreen}{AMD EPYC} \mtag{tagimproc}{In-process} \mtag{tagcodecin}{codec-inclusive} \textbf{Acknowledgment-path latency vs.\ offered load} (single-socket \texttt{r8a.metal-24xl}; measurement cores fully isolated, \S\ref{sec:runtime}). P50/P99 of end-to-end host-path response latency, order arrival to acknowledgment, codec-inclusive (inbound OUCH parse through OUCH/ITCH encode), per matcher segment, measured open-loop, median of ten runs, in the engine's production configuration and in its large-capacity configuration. In the production configuration the P50 stays under a microsecond through 104~M~msgs/s and the P99 through 86 (all ten runs at 82); the P99 stays under $10~\mu$s to the clean operating edge at \textbf{116~M~msgs/s}, and every run saturates at 118. The large-capacity configuration holds every run's P99 under a microsecond through \textbf{60~M~msgs/s} and reaches a clean operating edge of \textbf{124~M~msgs/s} (P99 9.77~$\mu$s); nine of ten runs saturate at 126. The frontiers at 60, 82, and 86~M~msgs/s come from two further batteries of ten runs per rate at 2~M~msgs/s steps around them, and they hold in both. \emph{sat} = the offered rate is not sustained; n/m = not measured.}
    \label{tab:latency_knee}
    \footnotesize
    \setlength{\tabcolsep}{4pt}%
    \begin{tabular}{r r r r r}
        \toprule
        \textbf{Offered load} & \multicolumn{2}{c}{\textbf{Production}} & \multicolumn{2}{c}{\textbf{Large-capacity}} \\
        \textbf{(M~msgs/s)} & \textbf{P50 (ns)} & \textbf{P99 (ns)} & \textbf{P50 (ns)} & \textbf{P99 (ns)} \\
        \midrule
        1   & 120 & 208 & 251 & 393 \\
        5   & 123 & 211 & 252 & 402 \\
        10  & 115 & 223 & 278 & 448 \\
        24  & 182 & 317 & 389 & 634 \\
        40  & 204 & 427 & 489 & 773 \\
        56  & 253 & 556 & 596 & 921 \\
        70  & 333 & 681 & 733 & 1{,}117 \\
        80  & 405 & 830 & 850 & 1{,}309 \\
        88  & 497 & 1{,}012 & 977 & 1{,}524 \\
        100 & 750 & 1{,}611 & 1{,}307 & 2{,}080 \\
        104 & 972 & 2{,}044 & 1{,}507 & 2{,}407 \\
        108 & 1{,}431 & 2{,}668 & 1{,}784 & 2{,}804 \\
        116 & 3{,}615 & 9{,}067 & 2{,}707 & 6{,}661 \\
        \midrule
        118 & \emph{sat} & \emph{sat} & 3{,}068 & 4{,}577 \\
        120 & \emph{sat} & \emph{sat} & 3{,}621 & 5{,}122 \\
        124 & n/m & n/m & 6{,}790 & 9{,}771 \\
        126 & n/m & n/m & \emph{sat} & \emph{sat} \\
        \bottomrule
    \end{tabular}
\end{table}

Every message type resolves well under a microsecond at these rates (Table~\ref{tab:op_latency}). Because the acknowledgment is generated on receipt, acceptance is confirmed before execution on the same ordered output stream; fills reach it at P50 185~ns, against 123~ns for acknowledgments (1~M~msgs/s). Cancel confirmation, the latency market makers watch and one rarely reported, is the fastest acknowledgment at \textbf{P50 104~ns / P99 155~ns}; the flow is cancel-dominated (742k cancel-acks per run against 59k fills), and every operation resolves under \textbf{320~ns} at P99. These are per-matcher-segment host-path figures that include OUCH parsing and OUCH/ITCH encoding; they are not comparable to matcher-only service times or to venue wire-to-wire numbers, which carry full production stacks and networks this measurement excludes.

\begin{table}[h]
    \centering
    \caption{\mtag{amdgreen}{AMD EPYC} \mtag{tagimproc}{In-process} \mtag{tagcodecin}{codec-inclusive} \textbf{Per-operation host-path latency} at 1 and 10~M~msgs/s, production configuration (\texttt{r8a.metal-24xl}; measurement cores fully isolated, \S\ref{sec:runtime}), P50\,/\,P99 in ns, median of three runs of ${\sim}2$M messages each. Every operation resolves under 320~ns at P99; cancel acknowledgment, the market-maker latency, is the fastest acknowledgment at both rates and both percentiles. The new-order-ack row reproduces Table~\ref{tab:latency_knee}'s production column within run-to-run scatter.}
    \label{tab:op_latency}
    \small
    \setlength{\tabcolsep}{4pt}%
    \begin{tabular}{l r r r}
        \toprule
        \textbf{Operation} & \textbf{1~M~msgs/s} & \textbf{10~M~msgs/s} & \textbf{Count/run} \\
        \midrule
        New-order ack & 123\,/\,214 & 119\,/\,225 & 944{,}931 \\
        Fill (trade)  & 185\,/\,292 & 183\,/\,318 & 58{,}994 \\
        Cancel ack    & \textbf{104\,/\,155} & \textbf{112\,/\,191} & 742{,}088 \\
        Cancel reject & 127\,/\,286 & 130\,/\,317 & 36{,}085 \\
        Modify ack    & 117\,/\,169 & 127\,/\,242 & 155{,}568 \\
        Modify reject & 133\,/\,299 & 139\,/\,317 & 7{,}425 \\
        \bottomrule
    \end{tabular}
\end{table}

\paragraph{What our engine adds to the pipeline.} To isolate our engine's cost, we rebuilt the same pipeline with the engine replaced by a null matcher that answers every message with one acknowledgment of the message's own type and keeps no book, no order index, and no executions; the ingress parse, the hand-offs between cores, and the OUCH/ITCH encoding are unchanged. Throughput and latency answer different questions here (Figure~\ref{fig:null_pipeline}). Throughput is the width of the pipe at its narrowest stage: in one interleaved battery on the production configuration, the pipeline with our engine's full price--time matching stays clean through \textbf{116~M~msgs/s}, 92\% of the 126~M~msgs/s the same pipeline sustains with no matching work at all. Latency is the time a message takes to pass through, and there our engine adds little: at 1--10~M~msgs/s, at most 17~ns to the median and 40--63~ns to the P99. The gap widens with load, and the curves part sharply as our pipeline nears its own edge, where the matching stage becomes the one that binds.

\begin{figure}[t]
    \centering
    \definecolor{nullgray}{HTML}{6B7380}
    \definecolor{fored}{HTML}{A11E2C}
    \begin{tikzpicture}[font=\sffamily]
    \begin{axis}[
        width=\columnwidth, height=6.6cm,
        xmin=0, xmax=136, ymin=88, ymax=11000,
        ymode=log,
        xtick={0,20,40,60,80,100,120},
        ytick={100,200,500,1000,2000,5000,10000},
        yticklabels={100\,ns,200\,ns,500\,ns,1\,$\mu$s,2\,$\mu$s,5\,$\mu$s,10\,$\mu$s},
        xlabel={Offered load (M~msgs/s)},
        ylabel={Response latency, new order to ack},
        label style={font=\sffamily\scriptsize, text=black!65},
        tick label style={font=\sffamily\scriptsize, text=black!55},
        axis lines=left, axis line style={black!35, line width=0.5pt, -},
        tick style={black!35, line width=0.4pt},
        ymajorgrids, grid style={black!9, line width=0.4pt},
        clip=false,
        legend columns=2,
        legend style={draw=none, fill=none, font=\sffamily\scriptsize, at={(0.0,1.03)}, anchor=south west, column sep=8pt, row sep=-1pt},
        legend image post style={xscale=0.7},
        legend cell align=left,
    ]
    % legend, drawn above the plot
    \addlegendimage{fored, line width=1.5pt, mark=*, mark size=1.1pt, mark options={draw=white, line width=0.3pt}}
    \addlegendentry{Flash One, P50}
    \addlegendimage{nullgray, line width=1.3pt, mark=*, mark size=1.0pt, mark options={draw=white, line width=0.3pt}}
    \addlegendentry{null matcher, P50}
    \addlegendimage{fored!50, line width=0.9pt}
    \addlegendentry{Flash One, P99}
    \addlegendimage{nullgray!50, line width=0.9pt}
    \addlegendentry{null matcher, P99}
    % the gap between the medians
    \addplot[name path=foMid, draw=none, forget plot] coordinates {(1,121) (5,122) (10,126) (24,188) (40,203) (56,253) (70,333) (80,405) (88,496) (96,632) (100,753) (104,984) (108,1444) (112,2080) (114,2627) (116,3687)};
    \addplot[name path=nuMid, draw=none, forget plot] coordinates {(1,113) (5,107) (10,109) (24,160) (40,177) (56,197) (70,242) (80,299) (88,351) (96,387) (100,407) (104,440) (108,470) (112,508) (114,538) (116,573)};
    \addplot[fored!13, draw=none, forget plot] fill between[of=foMid and nuMid];
    % the two clean operating edges
    \draw[fored!75, line width=0.7pt] (axis cs:116,88) -- (axis cs:116,11000);
    \draw[nullgray!75, line width=0.7pt] (axis cs:126,88) -- (axis cs:126,11000);
    \node[rotate=90, anchor=north east, font=\sffamily\scriptsize\bfseries, text=fored, inner sep=2pt] at (axis cs:116,10600) {our edge: 116};
    \node[rotate=90, anchor=north east, font=\sffamily\scriptsize\bfseries, text=nullgray, inner sep=2pt] at (axis cs:126,10600) {null edge: 126};
    % tails first, medians on top
    \addplot[fored!50, line width=0.9pt, forget plot] coordinates {(1,212) (5,210) (10,232) (24,320) (40,424) (56,552) (70,676) (80,821) (88,996) (96,1301) (100,1576) (104,2050) (108,2701) (112,3535) (114,4792) (116,9120)};
    \addplot[nullgray!50, line width=0.9pt, forget plot] coordinates {(1,172) (5,166) (10,169) (24,250) (40,286) (56,365) (70,487) (80,570) (88,644) (96,732) (100,810) (104,946) (108,977) (112,1110) (114,1219) (116,1386)};
    \addplot[nullgray, line width=1.3pt, mark=*, mark size=1.0pt, mark options={draw=white, line width=0.3pt}, forget plot] coordinates {(1,113) (5,107) (10,109) (24,160) (40,177) (56,197) (70,242) (80,299) (88,351) (96,387) (100,407) (104,440) (108,470) (112,508) (114,538) (116,573) (118,612) (120,659) (122,762) (124,870) (126,1300)};
    \addplot[fored, line width=1.5pt, mark=*, mark size=1.1pt, mark options={draw=white, line width=0.3pt}, forget plot] coordinates {(1,121) (5,122) (10,126) (24,188) (40,203) (56,253) (70,333) (80,405) (88,496) (96,632) (100,753) (104,984) (108,1444) (112,2080) (114,2627) (116,3687)};
    % reading guide, in the empty upper left
    \node[anchor=north west, align=left, font=\sffamily\scriptsize, text=black!80, inner sep=1pt] (guide) at (axis cs:3,7000) {The shaded gap is what our\\engine adds. At 1--10~M~msgs/s\\it is \textbf{$\leq$17 ns} at P50, \textbf{40--63 ns} at P99.};
    \draw[black!45, line width=0.4pt, -{Stealth[length=3pt,width=2.5pt]}] ([xshift=12pt]guide.south west) -- (axis cs:9,117);
    \end{axis}
    \end{tikzpicture}
    \caption{\mtag{amdgreen}{AMD EPYC} \mtag{tagimproc}{In-process} \mtag{tagcodecin}{codec-inclusive} \textbf{What our engine adds: the same wire pipeline with a null matcher and with our engine.} Response latency (log scale) against offered load, production configuration, both pipelines measured in one interleaved battery, median of ten runs per point. The shaded gap between the P50 curves is what our engine adds to the pipeline's latency. Each vertical line is a pipeline's clean operating edge, the highest rate at which all ten runs sustain the offered load: 116~M~msgs/s with our engine, 92\% of the 126~M~msgs/s the pipeline reaches with no matching work at all. P99 is drawn over the range both pipelines sustain.}
    \label{fig:null_pipeline}
\end{figure}

\paragraph{Robustness to order-identifier distribution.} \mtag{amdgreen}{AMD EPYC} \mtag{tagcabi}{C~ABI} Resolving a cancel by using the order identifier directly as a slot index is fast only while identifiers stay dense and small; under the sparse, high-entropy 64-bit identifiers a production venue can issue, that shortcut degrades into a scattered, cache-missing probe. Flash One does not take it: it resolves identifiers through an order-identifier index rather than by direct slot addressing, so it never relies on identifiers being dense. We verified this on the public harness with five remappings of the normal-day tape's identifiers that span how venues issue order IDs: maximally sparse, stride-aligned shard and venue tags, block-granted session clustering (the most production-realistic), 64-bit prefixed, and timestamp-like. Our single-core throughput stays at 127.1--127.9~M~msgs/s on every remapping, against 128.16 on the tape as generated (Table~\ref{tab:id_robustness}), and on every remapped tape our report stream is byte-identical to those of two independent conformers. The reported single-symbol numbers are a property of the engine, not an artifact of dense benchmark identifiers.

\begin{table}[h]
    \centering
    \caption{\mtag{amdgreen}{AMD EPYC} \mtag{tagcabi}{C~ABI} \textbf{Robustness to order-identifier allocation} (single-symbol throughput, M/s, public harness, normal-day tape; measurement cores fully isolated, \S\ref{sec:runtime}). The first row is the tape as generated (median of the battery); each remapped row is one run of the same tape with its identifiers remapped: maximally sparse (Knuth-scrambled), stride-aligned (the low-bits-constant shape of shard and venue tags), session-clustered (block-granted allocation as in per-session FIX counters and sharded sequencers, the most production-realistic), 64-bit prefixed, and timestamp-like. On every remapped tape our report stream is byte-identical to those of two independent conformers.}
    \label{tab:id_robustness}
    \small
    \begin{tabular}{l r}
        \toprule
        \textbf{Identifier pattern} & \textbf{Flash One (M/s)} \\
        \midrule
        Dense (as generated) & 128.16 \\
        Knuth-scrambled & 127.93 \\
        Stride-aligned ($\times$4096) & 127.68 \\
        Session-clustered & 127.62 \\
        64-bit prefixed & 127.11 \\
        Timestamp-like & 127.33 \\
        \bottomrule
    \end{tabular}
\end{table}

\subsubsection{Instance-level aggregate throughput}
\label{sec:instance_throughput}
\mbox{}\\
\mtag{blue!12}{ARM Graviton4} \mtag{tagimproc}{In-process}
The preceding experiments measure a single matcher core.  To characterize
the throughput ceiling of a node instance, we run the whole 96-core machine
(Section~\ref{sec:implementation}) as one matching instance: 47 independent
segments, each an engine instance on its own matcher and drainer core pair,
occupy the 94 cores outside the interrupt-housekeeping pair and service 213
symbols apiece, 10{,}011 in total.  Each segment replays its own normal-day tape
from the public harness's generator (about one million messages; 46.9 million per
run across the instance), with its orders assigned to its symbols by a
Zipf($\alpha=1.2$) draw, so every segment carries a few hot symbols (the hottest
takes about a quarter of its segment's orders, the top five about half) and a long
tail, and no single symbol concentrates the whole instance's load on one core;
within each symbol, order prices follow the tape's depth profile
(\S\ref{sec:methodology}).  All 47 segments start together on a shared clock, the
instance's throughput is every message of every segment over the time from the
first matcher's start to the last one's end, and every run is gated on each
segment's trade count and traded volume against a reference.

Under these conditions the engine reaches \textbf{${\sim}\,$1.3~billion
messages per second} across 10{,}011 symbols on a single instance
(1{,}279~M~msgs/s median over ten runs, 1{,}258--1{,}355).
\textbf{That figure is the machine's DRAM-bandwidth limit, not the engine's}.
Memory-controller counters, sampled every millisecond through the timed window and validated within 2\% against a STREAM-style calibration, put the instance's DRAM traffic, every thread included, at about 275 bytes per message, roughly constant as segments are added. That traffic averages 348--359~GB/s over the timed window and climbs to about 380~GB/s once all segments are running, this machine's practical ceiling for this mix of reads and writes: the same calibration reaches 366--376~GB/s writing and 404--424~GB/s reading or copying, against the 537~GB/s of the DDR5 specification. The aggregate flattens as DRAM traffic approaches that ceiling.
The cores can handle much more traffic given no restriction from the DRAM bandwidth.

The per-symbol distribution, Zipf-distributed flow across
symbols combined with the tape's depth profile within each symbol, concentrates
${\sim}$80\% of orders on the top ten price levels of each active
symbol, keeping each matcher's hot order-book nodes L1-resident.

For scale, the instance's ${\sim}\,1.3$~billion~msgs/s is over $45\times$ the CTA consolidated quote feed's 27~million~msgs/s
provisioned rate and more than $25\times$ the OPRA options feed's Q3~2024 one-second peak
(44.8~million~msgs/s)~\cite{ctaPlan,cboeOPRA2024}; those are consolidated market-data
dissemination feeds carrying different message types across the whole listed universe, not
single-node order-matching workloads, so the comparison sets magnitude, not a like-for-like
capability.

\subsection{Integration into Production Exchange Infrastructure}
\label{sec:integration}

We provide a license to practice our patent portfolio, with our proprietary optimizations delivered in binary form.
The matching engine architecture evaluated in this paper is designed
to integrate into existing exchange infrastructure as a drop-in library, not as a full replacement.

\textbf{Network ingress.}
The engine is a drop-in library of callable functions on the matching hot path; it consumes
parsed order descriptors, not raw packets, so existing exchange networking infrastructure
works with it unaltered. On AWS (Nitro/ENA), with kernel bypass (DPDK) on two ports of the
same instance and a paced client on the second, we measure a median NIC-to-NIC round trip of
${\sim}18~\mu$s, of which the engine's own packet-in to packet-out accounts for 0.41~$\mu$s;
the rest is the cloud fabric. Our P50 host-path latency of 115--123~ns
(\S\ref{sec:full_pipeline}) is a small fraction of any deployed round trip.

\textbf{Hardware acceleration path.}
Part of the multi-symbol cost is architectural: multiplexing many books through
one shared cache hierarchy degrades locality, on top of the per-symbol lookup each message pays; the per-message time rises from 8.3~ns at
one symbol to 10.1~ns at 10{,}000 (Table~\ref{tab:symbol_scaling}). The PIN maps directly onto dedicated on-chip memory: each book in
its own BRAM partition, no shared cache, matcher pipelines in parallel, so per-book cost can
stay flat until on-chip memory runs out. A report on an FPGA realization of the
architecture specified here is underway; the CPU results validate the
design that realization accelerates.

\begin{table*}[!t]
    \centering
    \caption{\mtag{amdgreen}{AMD EPYC} \mtag{tagcabi}{C~ABI} \mtag{tagcodecfree}{codec-free} \textbf{Response latency under burst load} (across the full message mix; wire codec excluded identically for all engines, so the comparison is protocol-fair), P50\,/\,P99 \textbf{in nanoseconds}, on the normal-trading-day tape (the representative operating point, \S\ref{sec:methodology}), measured open-loop and coordinated-omission-free from
    every message's scheduled arrival at the stated offered rate, on \texttt{r8a.metal-24xl}, measurement cores fully isolated (\S\ref{sec:runtime}), median over ten trials. Upper block: one book. Lower block: 10{,}000 books holding 132{,}000 resting orders on one matcher thread (\S\ref{sec:depth_sensitivity}), every engine's report stream byte-identical to the others'. ${\dagger}$~=~offered load exceeds the engine's
    sustainable throughput ($\rho>1$): the queue grows without bound, so the entry is the
    median delay accrued over the fixed ${\sim}$2M-message burst, not a convergent latency.
    ${}^{*}$~=~at the engine's saturation point: converged, but with a high variance.
    Every engine eventually diverges as $\rho\to1$. The engines shown are a representative high-throughput group (worst-case ceilings of ${\approx}\,5.5$--$8$~M~msgs/s in the Graviton4 census), chosen so the tested offered loads bracket their saturation points and isolate the objective latency effect of matcher saturation; the selection is not a judgment or verdict on any individual engine, and the engines' identities are anonymized as Projects~A--D${}^{\P}$.}
    \label{tab:burst_latency}
    \small
    \setlength{\tabcolsep}{4pt}%
    \begin{threeparttable}
    \begin{tabular}{l r r r r}
        \toprule
        \textbf{Engine} & \textbf{1~M~msgs/s} & \textbf{5~M~msgs/s} & \textbf{8~M~msgs/s} & \textbf{12~M~msgs/s} \\
        \midrule
        \multicolumn{5}{@{}l}{\emph{One book on a core}} \\
        \textbf{Flash One} & \textbf{156\,/\,222} & \textbf{147\,/\,213} & \textbf{145\,/\,211} & \textbf{145\,/\,253} \\
        Project~A & 216\,/\,358 & 203\,/\,368 & 219\,/\,460 & 1{,}528\,/\,243{,}961${}^{*}$ \\
        Project~B & 218\,/\,347 & 207\,/\,362 & 209\,/\,461 & 4{,}700{,}000\,/\,9{,}800{,}000${}^{\dagger}$ \\
        Project~C & 208\,/\,410 & 205\,/\,763 & 380\,/\,2{,}784 & 17{,}300{,}000\,/\,39{,}800{,}000${}^{\dagger}$ \\
        Project~D & 199\,/\,412 & 198\,/\,506 & 282\,/\,4{,}858 & 30{,}000{,}000\,/\,48{,}100{,}000${}^{\dagger}$ \\
        \midrule
        \multicolumn{5}{@{}l}{\emph{10k books, 132k resting orders on a core}} \\
        \textbf{Flash One} & \textbf{175\,/\,375} & \textbf{169\,/\,391} & \textbf{174\,/\,433} & \textbf{226\,/\,523} \\
        Project~A & 203\,/\,478 & 263\,/\,764 & 22{,}000{,}000\,/\,36{,}400{,}000${}^{\dagger}$ & 71{,}800{,}000\,/\,118{,}400{,}000${}^{\dagger}$ \\
        Project~B & 218\,/\,412 & 231\,/\,635 & 21{,}300{,}000\,/\,34{,}200{,}000${}^{\dagger}$ & 73{,}400{,}000\,/\,116{,}800{,}000${}^{\dagger}$ \\
        Project~C & 210\,/\,692 & 418\,/\,2{,}152 & 55{,}600{,}000\,/\,89{,}300{,}000${}^{\dagger}$ & 102{,}600{,}000\,/\,167{,}800{,}000${}^{\dagger}$ \\
        Project~D & 204\,/\,868 & 1{,}278\,/\,142{,}312${}^{*}$ & 76{,}900{,}000\,/\,117{,}900{,}000${}^{\dagger}$ & 127{,}900{,}000\,/\,198{,}900{,}000${}^{\dagger}$ \\
        \bottomrule
    \end{tabular}
    \begin{tablenotes}[flushleft]\footnotesize
        \item[${}^{\P}$] Baseline cells differ from those in earlier versions of this paper: they were re-measured on an optimized measurement vehicle, which lowers the baselines' P99 by about 0.8~$\mu$s; our engine's measurements are unaffected by the change.
    \end{tablenotes}
    \end{threeparttable}
\end{table*}

\section{Related Work}
\label{sec:related}

Public documentation from major venues emphasizes end-to-end latency and aggregate
capacity rather than per-symbol, per-core throughput under micro-bursts: Deutsche
B\"orse's T7 breaks latency down by pipeline stage without publishing a single-core
matching ceiling~\cite{db2025dynamics}, and Cboe's latency tooling reports port-to-port
statistics rather than core-saturation behavior~\cite{bats2011latency}. Open-source order
books almost universally adopt the same pattern, a linked list of orders per price level
indexed by a tree or array~\cite{howtohftLOB,quantstackLOB}. The public benchmark spans 246
of them across twenty-plus languages (\S\ref{sec:audited_field}), including the consensus
anchors (Liquibook~\cite{liquibook}, Exchange-core~\cite{exchangecore}, and
QuantCup~1~\cite{quantcupPTME}) and CoinTossX~\cite{jericevich2022cointossx}; none uses
contiguous priority-queue nodes or neighbor-aware tree operations.

\section{Matcher Throughput Is a Positional Good}
\label{subsec:commercial}
Where venues compete for order flow, sustainable per-symbol matcher throughput is not merely a systems goal; it behaves like a \emph{positional good}: a resource whose value depends on rank relative to competitors, not on its absolute level~\cite{hirsch1976}. The effect is most direct in U.S. equities, where Regulation NMS Rule~611 (the Order Protection Rule) prohibits any trading center from executing an order at a price inferior to a ``protected quotation'' displayed on another venue~\cite{secRule611,secRule611Memo}: a quotation is protected only if it is automated and immediately executable at an exchange's own best bid or offer; the National Best Bid or Offer (NBBO) is the best of those protected quotations across venues. Marketable flow gravitates to the venue displaying the NBBO: brokers route it there under their best-execution obligation, and Rule~611's trade-through protection bars other venues from filling it through that quote. A venue that is not at the NBBO tends to cede that flow to the venue that is. Among the eighteen active U.S. equity exchanges~\cite{cboeMarketShare} (three operated by Nasdaq, five by NYSE, four by Cboe, and six independents), the venue at the best price captures disproportionate marketable flow, and it is displaced the instant a competitor betters its price or is first to rebuild the quote after a sweep. What decides that capture is therefore not a venue's absolute matcher speed but its speed \emph{relative to} the competing venue's, and the difference bites under load: during a micro-burst, the matcher that clears the burst without queueing keeps its makers' quotes live at the NBBO while a saturating rival's quotes go stale. A matcher that is merely ``fast enough'' in isolation can still be positionally last.

\paragraph{The virtuous loop.} A venue whose matcher clears micro-bursts without queueing (\S\ref{subsec:microbursts}) holds response latency low and stable through the burst. Its market makers cancel and refresh quotes without waiting on a congested match loop, carry less stale-quote risk, and quote tighter and deeper; tighter quotes hold the NBBO more often, marketable flow routes to the venue under brokers' best-execution duty and Rule~611's trade-through protection, and that flow deepens the book and narrows spreads further: a liquidity externality that compounds the lead. The stakes are large: counting the off-exchange trades reported through the FINRA facility each group operates, Nasdaq's group handles the biggest share of trading (roughly 49\% of trading in NYSE-listed and 67\% of Nasdaq-listed names in March 2024, against the NYSE group's 33\% and 15\%)~\cite{nasdaqMarketData2024}, and shares of that magnitude are exactly what a durable edge in the NBBO competition can move.
\paragraph{The vicious loop.} A venue whose matcher saturates under the same flow inverts every step. The burst queues behind the serial match loop and response latency grows for as long as the burst lasts: by four to five orders of magnitude over the ${\sim}$2M-message bursts of Table~\ref{tab:burst_latency}. Resting quotes go stale before their makers can withdraw them, and stale quotes are exactly what latency-arbitrage traders pick off, so the venue's residual flow skews toward this \emph{toxic}, adversely-selected order flow while benign flow routes to the faster venue~\cite{budish2015batchauctions,aquilina2021armsrace}. Market makers absorbing those losses quote wider and thinner or withdraw; the venue loses the NBBO, loses the marketable flow that follows it, and its thinner book worsens its relative position further. Aquilina, Budish, and O'Neill estimate that market designs eliminating latency arbitrage would cut the market's cost of liquidity by 17\%~\cite{aquilina2021armsrace}.
Because both loops are self-reinforcing and driven by \emph{relative} matcher speed, a throughput edge does not merely improve a venue's own market: it transfers market share from slower competitors, and the size of the edge sets the size of the transfer. The base effect is measurable even without the NBBO amplifier: TSE's Arrowhead renewal in September 2015 approximately doubled order-processing capacity, halved order-response time, and added risk controls~\cite{fujitsu2015arrowhead}, after which effective spreads fell 6.45\% (particularly for large-cap, low-tick stocks) and marking-the-close manipulation declined 61\%~\cite{kemme2022tse}; over fiscal 2015 (April 2015 to March 2016), average daily cash-equity trading value rose 19.5\% year-on-year, lifting trading-services revenue from \textyen48.70 to \textyen52.47 billion, which JPX attributed to ``increases in trading of cash equities and derivatives''~\cite{jpxreport2016}. In fragmented U.S. equities, Rule~611 adds a positional amplifier: because a venue off the NBBO cedes marketable flow to the venue on it, a faster matching core does not merely improve its own market; it can absorb volume from slower competitors, converting a throughput edge into market-share gains.

\paragraph{The mechanism, measured.}\label{sec:burst_latency}
Throughput headroom is not an abstraction; under a long burst it is the difference between a
sub-microsecond response and a delay that keeps growing for as long as the burst lasts (Table~\ref{tab:burst_latency}). At 1~M~msgs/s every engine still answers in under a
microsecond at P99, on one book or on ten thousand; past its saturation point the queue grows for as long as the overload lasts (here ${\sim}$2M
messages, 170--250~ms), and response latency does not degrade gracefully: it climbs into
the tens and hundreds of millions of nanoseconds. With 10{,}000 books and 132{,}000 resting orders on one matcher thread, every other engine in Table~\ref{tab:burst_latency} has crossed that point by 8~M~msgs/s, and by 12~M~msgs/s their P99 response latencies land
between 117 and 199 \emph{million} nanoseconds, while our engine, which converges through
24~M~msgs/s in this scenario, holds \textbf{523~ns}. These are struct-to-struct figures, measured for every engine on the harness's C~ABI vehicle with the same tape and offered rates; the host-path figures of Table~\ref{tab:latency_knee} come from our own wire pipeline, a different measurement vehicle, and the two are not interchangeable.

The divergence at saturation is not linear, so a throughput margin that spans a burst's
peak rate becomes, for as long as the burst lasts, a latency gap of orders of
magnitude, which is what makes the throughput lead decisive, not incremental, where
venues compete for order flow.

\paragraph{Quantifying the stakes.}
The preceding data points allow back-of-the-envelope estimates of the
economic value at stake.  For exchanges: consolidated U.S. equity
volume averaged about \$1.1~trillion a day in 2025, about half of it on-exchange~\cite{cboe2025review},
so across approximately 250 trading days one percentage point of
on-exchange market share represents roughly
\$1.4~trillion a year in matched notional.  At the net capture U.S.
exchanges report after liquidity rebates (a few hundredths of a cent per
share~\cite{cboe4q2025}), that point is worth several million dollars a year in net
transaction revenue. For liquidity providers:
Aquilina, Budish, and O'Neill estimate that stale-quote sniping
extracts approximately \$5~billion per year from liquidity providers
across global equity markets alone (\$2.3--\$8.4~billion across
sensitivity analyses), imposing a roughly 0.5-basis-point tax on
trading volume~\cite{aquilina2021armsrace}.

\section{Conclusion}
\label{sec:conclusion}

We drove \textbf{247} matching engines (every common open-source FIFO implementation we could find, deduplicated, and our own) through one C~ABI on an identical, cancel-dominated workload calibrated to regulator-reported statistics. \textbf{160} reproduce a byte-identical correctness oracle. That oracle is the field's most exacting bug-finder: only \textbf{47} are correct as shipped, and we filed \textbf{181} GitHub issues upstream, \textbf{28} already fixed by their maintainers, none declined. Our engine leads all \textbf{160} by \textbf{${\sim}$95~M/s}, roughly $12.6\times$ the second best, on worst-case throughput (\S\ref{sec:audited_field}). A single core sustains \textbf{103.4~million} order messages per second (122.09~million on AMD~EPYC processors) worst-case, and in the engine's production configuration its wire pipeline keeps single-book P99 host-path latency, OUCH parsing and OUCH/ITCH encoding included, \textbf{under a microsecond through 82~M~msgs/s}, over two-thirds of the way to its clean operating edge (116~M~msgs/s); a 96-core node servicing over 10{,}000 symbols reaches \textbf{${\sim}\,$1.3~billion/s}, for scale, over $45\times$ the CTA consolidated quote feed's provisioned capacity (\S\ref{sec:instance_throughput}).

The result comes from treating per-symbol matching as a data-structure and cache-locality problem, not a networking one. Our evaluation isolates the matching pipeline itself, leaving networking, business logic, and regulatory checks to future work.

\paragraph{Intellectual property.}\label{sec:ip}
The order book architecture described in this paper is protected by a patent
portfolio covering the full stack, from software to hardware accelerator embodiments.  The portfolio
comprises multiple issued U.S. patents and pending international
patents through a PCT application.  All five U.S. applications
received first-action allowance from the USPTO, whose first-action allowance rate in
the relevant art unit is approximately 11\% (per USPTO statistics).  The open license under which this paper
is distributed covers its text and figures only.  It grants no license,
express or implied, to any issued patent or pending application covering
the architecture and its embodiments described herein;
implementing them requires a separate patent license.  For the avoidance of
doubt, this reservation does not extend to use of the benchmark harness: the
harness is MIT-licensed, and running it, including benchmarking your own
engine against the published workloads and reference hashes, requires no
separate license from Flash One.  We further reserve the right to
file continuation, continuation-in-part, divisional, or any similar
additional patent applications on the same subject matter in any
country in the world.  Both the issued claims and this ongoing
prosecution bear directly on any licensee's freedom-to-operate analysis,
and licensing inquiries may be directed to \texttt{jake@flash1.com}.

\section*{Engaging with Flash One}
\label{sec:engage}

The benchmark harness, baselines, and correctness oracle are open today, and the architecture behind our results is ready to evaluate against real workloads. Flash One is a patent IP licensing business, not a matching-engine vendor (we do not offer a full matching-engine product), so what follows is an invitation to license and evaluate the design, not a sales pitch for one. We invite two concrete next steps, matched to how you operate.

\textbf{Exchanges and trading venues.} We license the patented architecture behind these results, delivered as a drop-in binary library for you (or your technology partner) to integrate. The value is structural: wherever venues compete for order flow, as under Reg~NMS, matcher throughput behaves like a \emph{positional good} (\S\ref{subsec:commercial}); the venue that stays at the NBBO through micro-bursts captures flow from slower rivals, and the edge compounds. The licensed library drops into the per-symbol matching core (\S\ref{sec:integration}); at scale it reaches \textbf{${\sim}\,$1.3~billion} messages per second on a single 96-core node (${\sim}$\$1{,}630/month, 3-year reserved), for scale, over $45\times$ the CTA consolidated quote feed's provisioned capacity.

\textbf{HFT and market-making firms.} We are looking for strategic partnerships with HFT firms that actively participate in market making. You do not have to take our numbers on faith. Our public benchmark harness, baseline adapters, deterministic workload generator, and byte-identical reference hashes\footnote{\url{https://github.com/flash1-dev/matching-engine-benchmark}} let you reproduce the throughput of the best open-source engines on your own hardware and compare their numbers against ours. To put our own engine on the same footing, we will demonstrate it running that identical workload on identical hardware under a mutual non-disclosure and non-reverse-engineering agreement.

\bibliographystyle{ACM-Reference-Format}
\bibliography{acmart}

\end{document}